\documentclass[aip,cha,
reprint
]{revtex4-1}
\usepackage{graphicx}
\usepackage{amsmath,amssymb,amsthm}
\usepackage{dcolumn}
\usepackage{bm}
\usepackage{epsfig,color}
\usepackage{systeme}
\usepackage{multirow, comment, enumerate}
\usepackage{algorithm}
\usepackage{varwidth}
\usepackage{float}
\usepackage[noend]{algpseudocode}
\usepackage{setspace}
\usepackage{hyperref}
\usepackage{bigdelim}

\hypersetup{
colorlinks,
citecolor=blue,
linkcolor=blue,
urlcolor=blue
}

\allowdisplaybreaks
\theoremstyle{definition}
\newtheorem{theorem}{\textit{Theorem}}
\theoremstyle{definition}
\newtheorem{definition}{\textit{Definition}}
\theoremstyle{definition}
\newtheorem{lemma}{\textit{Lemma}}

\begin{document}
\title{An algebraic method to calculate parameter regions for constrained steady-state distribution in stochastic reaction networks}
\author{Tan Van Vu}
\email{tan@biom.t.u-tokyo.ac.jp}
\author{Yoshihiko Hasegawa}
\email{hasegawa@biom.t.u-tokyo.ac.jp}
\affiliation{Department of Information and Communication Engineering, Graduate School of Information Science and Technology, The University of Tokyo, Tokyo 113-8656, Japan}

\begin{abstract}
Steady state is an essential concept in reaction networks.
Its stability reflects fundamental characteristics of several biological phenomena such as cellular signal transduction and gene expression.
Because biochemical reactions occur at the cellular level, they are affected by unavoidable fluctuations. 
Although several methods have been proposed to detect and analyze the stability of steady states for deterministic models, these methods cannot be applied to stochastic reaction networks.
In this paper, we propose an algorithm based on algebraic computations to calculate parameter regions for constrained steady-state distribution of stochastic reaction networks, in which the means and variances satisfy some given inequality constraints.
To evaluate our proposed method, we perform computer simulations for three typical chemical reactions and demonstrate that the results obtained with our method are consistent with the simulation results.
\end{abstract}

\maketitle

\noindent\textbf{Many biological phenomena like cellular signal transduction and gene expression can be described by stochastic reaction networks.
It has been known that these systems function robustly at the steady state in the presence of noise.
Therefore, it is pertinent to ask which conditions of parameters define such stable operating regime of the system.
This kind of information provides insights into understanding the underlying mechanism of the system, and particularly into the design processes of stochastic biocircuits.
Here, we propose an algebraic method to calculate the parameter regions in which the means and variances at the steady state satisfy some given constraints.
These constraints can be, for example, an upper bound of variances, coefficients of variation, or Fano factors to control fluctuations.
We show in the experiments that our approach gives results comparable with stochastic simulations.
As all computations are symbolic, our method does not require the prior knowledge of the parameters, which are often unavailable in biological systems.
The constraints can also be intentionally added to obtain the conditions of parameters under which the system is brought into a desired steady state.
}

\section{Introduction}
Biochemical reaction networks are mathematical models used for describing biological processes, such as signal transduction and gene expression, at the cellular level \cite{Thattai.2001,Shibata.2005}.
Since most of the biological processes that constitute an organisms' activity are either in or moving toward a steady state, elucidating the stability of the steady state provides an understanding of the behavior of the biochemical processes \cite{Fichera.1977,Angeli.2004}.
Although a great deal of effort has been put into the formulation and analysis of reaction networks, detecting and analyzing steady states of reaction networks remain challenging \cite{Karin.2005,Gheorghe.2009,Mincheva.2012}. 
Two commonly used models for reaction networks are continuous deterministic models and discrete stochastic models. In deterministic models, no randomness is involved and all state variables are predictable. 
Deterministic models are often described by a set of ordinary differential equations and are most appropriate when the molecule numbers of all reactant species are sufficiently large that underlying fluctuations can be ignored.
However, in biological processes, noise is unavoidable and plays functional roles such as noise-induced bistability and oscillation \cite{Domijan.2009,Tsimring.2014,Hasegawa.2014}. 
Low molecule numbers of only a few reactant species can lead to significant fluctuations. 
In such cases, deterministic models fail to accurately depict the dynamics of the system; therefore, stochastic models are necessary.
Consequently, steady state in stochastic models becomes a distribution rather than a fixed point as in deterministic models.
Previous studies have shown that the stochastic dynamics of a well-mixed chemically reacting system can be accurately modeled by the chemical master equation (CME) \cite{Donald.1967,Gillespie.1992}. 
In most cases, the CME has not been analytically solved. Hence numerical computations such as stochastic simulation algorithms \cite{Gillespie.1977,Gibson.2000,Gillespie.2001}, finite state projection method \cite{Brian.2006,Gupta.2017}, and quantized tensor trains \cite{Vladimir.2014} are often conducted.
Although these numerical methods can help us obtain the distribution of the network with a particular parameter value, it is intractable to apply these methods to calculate the parameters that yield a desired steady state, in which the means and variances satisfy some given constraints.

Over the past few years, there have been many attempts to apply algebraic methods to the analysis of reaction networks \cite{Wang.2005,Manrai.2008,Gheorghe.2009,Thomson.2009,Martinez.2010,Loriaux.2013,SiegalGaskins.2015,Otero.2017,Jose.2017}.
The biggest advantage of the algebraic method is that it does not require knowledge of parameter values. 
In algebraic computations, the parameters of reaction networks are treated as symbolic quantities rather than as numbers. 
Moreover, at the steady state, rate equations of the deterministic model form an algebraic variety that can be studied using algebraic geometry.
Mart{\'i}nez \textit{et al.} \cite{Martinez.2010} proposed a procedure to locate the steady states of reaction networks from an algebraic geometry method-derived formula.
By computing Gr{\"o}bner basis \cite{Cox.2007} of rate equations, one can derive the bifurcation diagram of the reactant concentration at the steady state in terms of a specific parameter. 
In \cite{Wang.2005}, by exploiting exact symbolic computation, an approach is presented for analyzing the stability of a large class of biological networks that modeled as autonomous systems of differential equations. 
Siegal-Gaskins \textit{et al.} \cite{SiegalGaskins.2015} applied Sturm's theorem to the analysis of bistable biological circuits. 
All of these works have been applied to deterministic models.
One of the differences between stochastic and deterministic models is the concept of the steady state.
In deterministic models, the system may have multiple steady states which are fixed points.
On the contrary, in stochastic models, the system always holds a unique steady state depending on the initialization \cite{Gupta.2014,Gupta.2018}, which is a probability distribution.

In the present paper, we propose an algebraic method to calculate parameter regions in which the steady-state distribution of the reaction networks satisfies given constraints of means, variances, or fluctuation characteristics like coefficient of variation and Fano factor \cite{Thattai.2001,Ozbudak.2002,Mads.2003}.
The procedure of our method is as follows.
First, we compute closed-moment equations of reaction networks via moment closure approximations or linear noise approximation. 
At the steady state, moments are considered to be time-invariant and moment equations form an algebraic variety whose solutions hold the information of moments at the steady state. 
Since moment equations are obtained through approximation, there is the possibility that physically inappropriate values of moments will be included in the solutions.
To eliminate these inadmissible solutions, inequality constraints, positivity of both means and variances, and the upper bound of variances are added \cite{Juan.2017}.
Eventually, a system of multivariate polynomials that contains equations and inequations is obtained. 
Finally, we apply an algebraic method to compute the conditions of parameters such that the polynomial system has exactly one solution.
We demonstrate the validity of our method on three well-known reaction network models: a gene regulatory system, two-component Michaelis--Menten enzyme reactions, and a Brusselator model.
We perform stochastic simulations to sample the desired regions of parameters.
Although the ranges of parameters obtained with our method are approximate (for nonlinear systems), experiments show that the results of our method agree with the simulation.
The results of our proposed method provide an insight into the dynamic behavior of the reaction network at the steady state.
The volume of the space of admissible parameters can be considered as a quantity, which indicates the robustness of the system.
Moreover, with the flexibility of adding constraint conditions, our method can be used as a tool to explore the parameter configurations, which satisfy requirements in the design process of stochastic biocircuits \cite{Sakurai.2017}.

\section{Preliminaries}\label{sect.pre}
\subsection{Reaction networks}
We consider a general reaction network with $N$ reactant species $X_1,\dots,X_N$ interacting through $M$ reaction channels $\mathcal{C}_1,\dots,\mathcal{C}_M$ inside a cell with fixed volume $\Omega$. 
The system is assumed to be well-mixed. The reaction channel $\mathcal{C}_j~(1\leq j\leq M)$ is of the type:
\begin{equation*}
a_{1j}X_1+\dots+a_{Nj}X_N\xrightarrow{k_j}b_{1j}X_1+\dots+b_{Nj}X_N,
\end{equation*}
where $a_{ij},b_{ij}\in\mathbb{N}_{\geq 0}$ are the stoichiometric coefficients and $k_j\in\mathbb{R}_{>0}$ is the macroscopic rate of reaction. 
The state of system is fully determined by the vector of molecule numbers of each species, $\boldsymbol{n}=(n_1,\dots, n_N)$, where $n_i\in\mathbb{N}_{\geq 0}$ is the molecule number of the species $X_i$. 
The CME describing the time evolution of the system is given by
\begin{equation}\label{eq.cme}
\frac{\partial P(\boldsymbol{n},t)}{\partial t}=\sum_{j=1}^{M}(f_{j}(\boldsymbol{n}-\boldsymbol{V}_{j})P(\boldsymbol{n}-\boldsymbol{V}_{j},t)-f_{j}(\boldsymbol{n})P(\boldsymbol{n},t)),
\end{equation}
where $\boldsymbol{V}=[b_{ij}-a_{ij}]\in\mathbb{Z}^{N\times M}$ is a stoichiometric matrix, $\boldsymbol{V}_{j}$ denotes the $j^{\mathrm{th}}$ column of matrix $\boldsymbol{V}$, $P(\boldsymbol{n},t)$ is the probability that the system will be in state $\boldsymbol{n}$ at time $t$, and $f_j(\boldsymbol{n})$ represents the propensity function to account for the transition from a given state $\boldsymbol{n}$ to any other state in the reaction channel $\mathcal{C}_j~(1\leq j\leq M)$. 
Under the assumption of mass-action kinetics, the propensity function $f_j(\boldsymbol{n})$ has the following form:
\begin{equation}
f_j(\boldsymbol{n})=\Omega^{1-\sum_{i=1}^{N}a_{ij}}k_j\prod_{i=1}^{N}\frac{n_i!}{(n_i-a_{ij})!}.
\end{equation}
The solution of CME completely describes the stochastic dynamics of the system. 
However, in most cases, this differential equation is extremely difficult to solve explicitly. 
Consequently, stochastic simulation algorithms, such as the Gillespie algorithm \cite{Gillespie.1977} or its modifications \cite{Gibson.2000,Gillespie.2001}, are often used to simulate the dynamics of the system. 
Although stochastic simulation can exactly describe the stochastic evolution of the system, these methods are very computationally expensive when the number of reactant species is large. 
In these cases, the precision of the stochastic simulation is often sacrificed for faster, yet more approximate, methods. 
Various numerical and analytical methods have been proposed to approximately solve the CME, e.g., approximations of the CME solution by solving a truncated version of the Markov process \cite{Brian.2006,Vladimir.2014}, linear noise approximation \cite{Johan.2003,Scott.2006,Kampen.2007,Wallace.2012}, moment closure methods \cite{Carlos.2007,Ferm.2008,Chang.2009,Gillespie.2009,Ramon.2012,Angelique.2013}, and chemical Langevin equation treatments \cite{Gillespie.2000}.

\subsection{Algebraic preliminaries}
Let $\mathbb{K}$ be an arbitrary field and $\mathbb{K}[x_1,\dots,x_n]$ be the ring of multivariate polynomials over $\mathbb{K}$ with variables $x_1\prec x_2\prec\dots\prec x_n$, where $\prec$ denotes the ascending order of variables. 
$\boldsymbol{x}$ and $\boldsymbol{x}_m~(m<n)$ are used to denote $x_1,\dots,x_n$ and $x_1,\dots,x_m$, respectively. 
For any polynomial $P \in \mathbb{K}[\boldsymbol{x}]$ and variable $x_m$, $P$ can be viewed as a univariate polynomial in $x_m$ over $\mathbb{K}[\boldsymbol{x}_{m-1},x_{m+1},\dots,x_n]$. 
$\mathrm{deg}(P,x_m)$ denotes the degree of $P$ in $x_m$ and $\mathrm{lc}(P,x_m)$ represents the leading coefficient of $P$ with respect to (w.r.t.) $x_m$. 
For convenience, we define $\mathrm{deg}(0,x_m)\triangleq -1$. 
The largest variable effectively appearing in $P$ is called the leading variable of $P$ and is denoted by $\mathrm{lv}(P)$. 
If $\mathrm{lv}(P)=x_i$ then the initial and reductum of $P$ will be defined as follows:
\begin{equation*}
\mathrm{ini}(P)\triangleq \mathrm{lc}(P,x_i),~\mathrm{red}(P)\triangleq P-\mathrm{lc}(P,x_i)x_i^{\mathrm{deg}(P,x_i)}.
\end{equation*}
For any two nonzero polynomials $P,Q\in\mathbb{K}[\boldsymbol{x}]$ with $\mathrm{deg}(P,x_m)=n_1$ and $\mathrm{deg}(Q,x_m)=n_2>0$, the pseudo-division algorithm computes two polynomials $S,R\in \mathbb{K}[\boldsymbol{x}]$ such that $I^rP=SQ+R$, where $I=\mathrm{lc}(Q,x_m),~r=\max(n_1-n_2+1,0),~\mathrm{deg}(S,x_m)=\max(n_1-n_2,-1)$, and $\mathrm{deg}(R,x_m) < n_2$. 
The polynomials $S$ and $R$ are called the \textit{pseudo-quotient} and \textit{pseudo-remainder} of $P$ w.r.t. $Q$ in $x_m$ and denoted by $\mathrm{pquo}(P,Q,x_m)$ and $\mathrm{prem}(P,Q,x_m)$, respectively. 
If $\mathrm{lv}(Q)=x_i$ then $\mathrm{prem}(P,Q)\triangleq \mathrm{prem}(P,Q,x_i),~\mathrm{pquo}(P,Q)\triangleq \mathrm{pquo}(P,Q,x_i)$. $\mathrm{gcd}(P,Q,x_m)$ denotes the \textit{greatest common divisor} of $P\neq 0$ and $Q\neq 0$ w.r.t. $x_m$.

For any two polynomial sets $\mathcal{P}=\{P_1,P_2,\dots,P_{N_p}\}$ and $\mathcal{Q}=\{Q_1,Q_2,\dots,Q_{N_q}\}$, the expressions $\mathcal{P}=0, \mathcal{Q}\neq 0, \mathcal{Q}>0$ correspond to $\{P_1=0,P_2=0,\dots,P_{N_p}=0\}$, $\{Q_1\neq 0,Q_2\neq 0,\dots,Q_{N_q}\neq 0\}$ and $\{Q_1>0,Q_2>0,\dots,Q_{N_q}>0\}$, respectively.
$\mathrm{Zero}(\mathcal{P})$ denotes the set of all common real zeros of the polynomials in $\mathcal{P}$. Additionally, we use the following notation:
\begin{align*}
\mathrm{Zero}(\mathcal{P}\setminus\mathcal{Q})&\triangleq\{\boldsymbol{x}|\boldsymbol{x}\in\mathrm{Zero}(\mathcal{P})\setminus\mathrm{Zero}(\mathcal{Q})\},\\
\mathrm{Zero}(\mathcal{P},\mathcal{Q}>0)&\triangleq\{\boldsymbol{x}|\boldsymbol{x}\in\mathrm{Zero}(\mathcal{P}),Q(\boldsymbol{x})>0~\forall Q\in\mathcal{Q}\}.
\end{align*}
For any set $\mathcal{S}$, $\left|\mathcal{S}\right|$ denotes the cardinality of the set $\mathcal{S}$, i.e., the number of elements of $\mathcal{S}$.

\subsubsection{Resultant and subresultant}
Suppose we are given two polynomials $P(x),Q(x)$ with $n_1 \geq n_2>0$ as follows:
\begin{align*}
P(x)&=p_0x^{n_1}+p_1x^{n_1-1}+\cdots+p_{n_1-1}x+p_{n_1} &(p_0 \neq 0)\\
Q(x)&=q_0x^{n_2}+q_1x^{n_2-1}+\cdots+q_{n_2-1}x+q_{n_2} &(q_0 \neq 0).
\end{align*}
$(n_1+n_2)\times (n_1+n_2)$ Sylvester matrix $\boldsymbol{S}$ of $P$ and $Q$ has the following form:
\begin{equation*}
\begin{array}{c@{}c}
\boldsymbol{S} = 
\begin{pmatrix}
p_{0} & p_{1} & \cdots & p_{n_1} & & \\[0.02em]
& \cdots  & \cdots  & \cdots & \cdots \\[0.02em]
& & p_{0} & p_{1} & \cdots & p_{n_1} \\[0.02em]
q_{0} & q_{1} & \cdots & q_{n_2} & & \\[0.02em]
& \cdots  & \cdots  & \cdots & \cdots\\[0.02em]
& & q_{0} & q_{1} & \cdots & q_{n_2}
\end{pmatrix}
&
\begin{array}{l}
\\[-9mm] \rdelim\}{2}{0mm}[$n_2$] \\[9mm]\rdelim\}{2}{0mm}[$n_1$]
\end{array}
\end{array}\quad .
\end{equation*}
\begin{definition}
Resultant of two polynomials $P$ and $Q$ w.r.t. $x$, which is denoted by $\mathrm{res}(P,Q,x)$, is the determinant of $\boldsymbol{S}$.
\end{definition}
If $P,Q\in\mathbb{K}[x]$, then $\mathrm{res}(P,Q,x)=0$ if and only if $P$ and $Q$ have common non-constant factors in $\mathbb{K}[x]$ \cite{Cox.2007}. 
Moreover, if $\mathrm{res}(P,P',x)\neq 0$ then the polynomial $P(x)$ does not hold multiple roots.

Let $\boldsymbol{S}_{i'i}$ be the submatrix of $\boldsymbol{S}$ obtained by deleting the last $i$ of the $n_2$ rows of $P$ coefficients, the last $i$ of the $n_1$ rows of $Q$ coefficients and the last $2i+1$ columns, excepting column $n_1+n_2-i'-i$, for $0\leq i'\leq i\leq n_2$. 
The polynomial $S_{i}(x)=\sum_{i'=0}^{i}\mathrm{det}(\boldsymbol{S}_{i'i})x^{i'}$ is then called the $i^{\mathrm{th}}$ \textit{subresultant} of $P$ and $Q$ w.r.t. $x$ for $0\leq i\leq n_2$. 
If $n_1>n_2+1$, the definition of $i^{\mathrm{th}}$ subresultant $S_i(x)$ of $P$ and $Q$ w.r.t. $x$ is extended as follows:
\begin{equation*}
S_{n_2}(x)=q_0^{n_1-n_2-1}Q,~ S_i(x)=0,~ n_2<i<n_1-1.
\end{equation*}
$S_i$ is said to be \textit{defective} of degree $r$ if $\mathrm{deg}(S_i,x)=r<i$, and \textit{regular} otherwise.
\begin{definition}
Let $P, Q \in \mathbb{K}[x]$ be two polynomials with $n_1=\mathrm{deg}(P,x)\geq\mathrm{deg}(Q,x)=n_2>0$ and set
\begin{equation*}
\bar{n}=
\begin{cases}
n_1-1, &\quad \text{if } n_1>n_2\\
n_2, &\quad \text{otherwise.}
\end{cases}
\end{equation*}
Let $S_{\bar{n}+1}=P,S_{\bar{n}}=Q$ and $S_i$ be the $i^{\mathrm{th}}$ subresultant of $P$ and $Q$ w.r.t. $x$ for $0\leq i<\bar{n}$, then the sequence of polynomials $S_{\bar{n}+1},S_{\bar{n}},\dots,S_0$ is called the \textit{subresultant chain} of $P$ and $Q$ w.r.t. $x$.
\end{definition}
There is an effective method for constructing subresultant chains by means of pseudo-division \cite{Mishra.1993}.
\begin{definition}
Let $S_{\bar{n}+1},S_{\bar{n}},\dots,S_0$ be the subresultant chain of $P$ and $Q$ w.r.t. $x$. The sequence of regular subresultants $S_{d_2},\dots,S_{d_r}$ is called the \textit{subresultant regular subchain} (s.r.s.) of $P$ and $Q$ w.r.t. $x$ if
\begin{itemize}
\setlength\itemsep{0.2em}
\item[1.] $\bar{n}+1=d_1>d_2>\dots>d_r\geq 0$.
\item[2.] $S_{d_{i'}}$ is regular for all $2\leq i'\leq r$ and $S_i$ is defective for all $i \in \{0,\dots,\bar{n}\}\setminus\{d_2,\dots,d_r\}$.
\end{itemize}
\end{definition}
The s.r.s $S_{d_2},\dots,S_{d_r}$ is renamed $H_2,\dots,H_r$ in Algorithm \ref{alg.seqdecomp}. 
For any two polynomials $P,Q\subset\mathbb{Z}[x_1,\dots,x_n]$ with $\mathrm{deg}(P,x_n)\geq\mathrm{deg}(Q,x_n)$, the s.r.s $\{H_i\}_{i=2}^{r}$ of $P,Q$ provides an efficient way to calculate $\mathrm{gcd}(P,Q,x_n)$ without computing multiple gcds. 
Moreover, this s.r.s can be exploited to calculate $\mathrm{Zero}(\{P,Q\}\setminus I)$ as follows \cite{Wang.1998}:
\begin{equation*}
\mathrm{Zero}(\{P,Q\}\setminus I)=\bigcup_{i=2}^{r}\mathrm{Zero}(\{H_i,I_{i+1},\dots,I_r\}\setminus\{I,I_{i}\}),
\end{equation*}
where $I=\mathrm{lc}(Q,x_n),I_i=\mathrm{lc}(H_i,x_n)$ for all $i=2,\dots,r$.
\begin{definition}
A polynomial set $\mathcal{T}=\{T_1,T_2,\dots,T_n\} \subset \mathbb{K}[\boldsymbol{x}]$ is called a triangular set if
\begin{itemize}
\setlength\itemsep{0.2em}
\item[1.] $\mathcal{T}\cap\mathbb{K}=\varnothing$, i.e., $\mathcal{T}$ does not contain any constant polynomial.
\item[2.] $\mathrm{lv}(T_{i'})\prec\mathrm{lv}(T_i)$ for all $1\leq i'<i\leq n$.
\end{itemize}
\end{definition}
We use the following notations:
\begin{align*}
\mathcal{T}^{(m)}&\triangleq\{T\in\mathcal{T}\mid\mathrm{lv}(T)\preceq x_m\},\\
\mathcal{T}^{\langle m\rangle}&\triangleq\{T\in\mathcal{T}\mid\mathrm{lv}(T)=x_m\}.
\end{align*}
The pseudo-remainder $\mathrm{prem}(P,\mathcal{T})$ of any polynomial $P\in\mathbb{K}[\boldsymbol{x}]$ w.r.t. $\mathcal{T}$ is defined recursively as
\begin{equation*}
\mathrm{prem}(P,\mathcal{T})\triangleq\mathrm{prem}\left(\mathrm{prem}(P,T_n,\mathrm{lv}(T_n)),\mathcal{T}^{(n-1)}\right),
\end{equation*}
where $\mathrm{prem}(P,\varnothing)\triangleq P$. 
Similarly, we define $\mathrm{res}(P,\mathcal{T})\triangleq\mathrm{res}\left(\mathrm{res}(P,T_n,\mathrm{lv}(T_n)),\mathcal{T}^{(n-1)}\right)$, where $\mathrm{res}(P,\varnothing)\triangleq P$. 
If $\mathrm{res}(P,\mathcal{T})\neq 0$ then $\mathrm{Zero}(\{P\})$ and $\mathrm{Zero}(\mathcal{T})$ have no elements in common, i.e., $P$ has no common solution with the system $\{\mathcal{T}=0\}$.
\begin{definition}
A polynomial set $\mathcal{T}=\{T_1,T_2,\dots,T_n\} \subset \mathbb{K}[\boldsymbol{x}]$ is called a regular set if
\begin{itemize}
\setlength\itemsep{0.2em}
\item[1.] $\mathcal{T}$ is a triangular set.
\item[2.] $\mathrm{res}(\mathrm{ini}(T_i),\{T_1,\dots,T_{i-1}\}) \neq 0$ for all $1\leq i\leq n$, i.e., the leading coefficient of $T_i$ is non-zero when $x_1=\bar{x}_1,\dots,x_{i-1}=\bar{x}_{i-1}$ are substituted. Here, $\{\bar{x}_1,\dots,\bar{x}_{i-1}\}$ is a solution of $\mathcal{T}^{(i-1)}=0$.
\end{itemize}
\end{definition}
For any two polynomial sets $\mathcal{P}$ and $\mathcal{Q}$, if $\mathcal{P}$ is a regular set then we call $\{\mathcal{P},\mathcal{Q}\}$ a regular system.

\section{Methods}\label{sect.met}
\subsection{Moment equations}
Because obtaining analytical probability distributions that satisfy the CME is very difficult in most cases, their moments (e.g., mean and variance) are often computed to capture the behaviors of stochastic reaction networks. 
For obtaining moment equations, two methods are commonly used: moment closure approximation and linear noise approximation (LNA). 
For linear systems where all propensity functions of the system are linear, explicit expressions for moments can be obtained. 
However, for nonlinear systems, this is not the case because each moment depends on higher moments, resulting in an infinite hierarchy of moment equations. 
For that reason, several approximations have been proposed to obtain closed moment equations. 
These approaches approximate higher order moments by lower order ones assuming probability distributions (normal distribution \cite{Carlos.2007,Ramon.2012,Angelique.2013}, lognormal distribution \cite{Matt.2000,Nasell.2003}, Poisson distribution \cite{Nasell.2003}) or totally ignore the higher order moments (central moment neglect \cite{Ferm.2008,Gillespie.2009,Chang.2009}). 
If moment equations are truncated at the second order, i.e., only means and variances (covariances) are considered, approximations based on normal distribution and central moment neglect become equivalent. 
It has been reported that the moment closure approximation based on normal distribution is advantageous over others, which provides a larger range of parameter space where it gives physically meaningful results \cite{David.2015}. 
For that reason, we adopt here the approximations based on normal distribution and LNA to obtain moment equations.
We stress that in our method, the approximation scheme can be replaced by other better ones if available without affecting the method pipeline.

\subsubsection{Approximation based on normal distribution}
We define the first two moments, i.e., means $\boldsymbol{\mu}$ and variances $\boldsymbol{\sigma}$, as follows:
\begin{align}
\mu_{i}&=\langle n_i\rangle=\sum_{\boldsymbol{n}}n_iP(\boldsymbol{n},t),\\
\sigma_{ii'}&=\langle (n_i-\mu_{i})(n_{i'}-\mu_{i'})\rangle\nonumber\\
&=\sum_{\boldsymbol{n}}(n_i-\mu_{i})(n_{i'}-\mu_{i'})P(\boldsymbol{n},t),
\end{align}
and set all central moments above order two to zero. After some transformations (see Appendix \ref{app.moment.equations}), the equations of the first two moments are
\begin{equation}
\label{eq.moments}
\begin{aligned}
\frac{d\mu_{i}}{dt}&=\sum_{j=1}^{M}V_{ij}\Bigg(f_j(\boldsymbol{\mu})+\frac{1}{2}\sum_{h,l}\frac{\partial^2 f_j(\boldsymbol{\mu})}{\partial n_h\partial n_l}\sigma_{hl}\Bigg)\\
\frac{d\sigma_{ii'}}{dt}&=\sum_{j=1}^{M}\Bigg(V_{ij}V_{i'j}\Bigg(f_j(\boldsymbol{\mu})+\frac{1}{2}\sum_{h,l}\frac{\partial^2 f_j(\boldsymbol{\mu})}{\partial n_h\partial n_l}\sigma_{hl}\Bigg)\\
&+V_{i'j}\sum_{l}\frac{\partial f_j(\boldsymbol{\mu})}{\partial n_l}\sigma_{il}+V_{ij}\sum_{l}\frac{\partial f_j(\boldsymbol{\mu})}{\partial n_l}\sigma_{i'l}\Bigg).
\end{aligned}
\end{equation}
At the steady state, all moments are time-invariant and moment equations form an algebraic variety as below:
\begin{align}
&\sum_{j=1}^{M}V_{ij}\Bigg(f_j(\boldsymbol{\mu})+\frac{1}{2}\sum_{h,l}\frac{\partial^2 f_j(\boldsymbol{\mu})}{\partial n_h\partial n_l}\sigma_{hl}\Bigg)=0,~\forall i=1,\dots,N,\nonumber\\
&\sum_{j=1}^{M}\Bigg(V_{ij}V_{i'j}\Bigg(f_j(\boldsymbol{\mu})+\frac{1}{2}\sum_{h,l}\frac{\partial^2 f_j(\boldsymbol{\mu})}{\partial n_h\partial n_l}\sigma_{hl}\Bigg)\label{eq.moments.steady.states}\\
&+V_{i'j}\sum_{l}\frac{\partial f_j(\boldsymbol{\mu})}{\partial n_l}\sigma_{il}+V_{ij}\sum_{l}\frac{\partial f_j(\boldsymbol{\mu})}{\partial n_l}\sigma_{i'l}\Bigg)=0,\nonumber\\
&\hspace{60mm}\forall i,i'=1,\dots,N.\nonumber
\end{align}

\subsubsection{Linear noise approximation}
The CME can be simplified in the LNA through $\Omega$-expansion \cite{Kampen.2007}. 
To derive the LNA, we approximate the CME by a Taylor expansion of the state variables $\boldsymbol{n}$ around the mean values $\boldsymbol{\mu}$ of the concentrations.
The fluctuations in $\boldsymbol{n}$ are assumed to be of the order of $\mathcal{O}(\Omega^{1/2})$. Then $\boldsymbol{n}$ can be expressed as
\begin{equation}
\boldsymbol{n}=\Omega\boldsymbol{\mu}+\sqrt{\Omega}\boldsymbol{\xi},
\end{equation}
where $\boldsymbol{\xi}$ are the fluctuating variables.
The CME can be rewritten as
\begin{equation}
\frac{\partial P(\boldsymbol{n},t)}{\partial t}=\sum_{j=1}^{M}\left[\left(\prod_{i=1}^{N}\mathbb{E}_i^{-V_{ij}} \right)-1\right]f_{j}(\boldsymbol{n},\Omega)P(\boldsymbol{n},t),\label{eq.master.operator}
\end{equation}
where $\mathbb{E}_i^v$ is an operator that replaces $n_i$ by $n_i+v$. The probability distribution $P(\boldsymbol{n},t)$ can be replaced by the distribution of fluctuations $\Pi(\boldsymbol{\xi},t)$ as $P(\boldsymbol{n},t)=\Omega^{-N/2}\Pi(\boldsymbol{\xi},t)$. By using the Taylor expansion, we can approximate the operator $\mathbb{E}_i^v$ as follows:
\begin{equation}
\mathbb{E}_i^{v}\approx\left[1+\frac{v}{\sqrt{\Omega}}\frac{\partial}{\partial\xi_i}+\frac{v^2}{2\Omega}\frac{\partial^2}{\partial\xi_i^2}+\dots\right].\label{eq.operator.approximation}
\end{equation}
Substituting the approximation in Eq.~\eqref{eq.operator.approximation} into Eq.~\eqref{eq.master.operator} and collecting the terms of the order of $\mathcal{O}(\Omega^0)$ and $\mathcal{O}(\Omega^{-1/2})$, we obtain the rate equations for the means and the Fokker-Planck equation for the fluctuations as
\begin{align}
\frac{d\boldsymbol{\mu}}{dt}&=\boldsymbol{V}\boldsymbol{\nu},\\
\frac{\partial\Pi}{\partial t}&=-\sum_{i,j}\Gamma_{ij}\partial_i(\xi_i\Pi)+\frac{1}{2}\sum_{i,j}D_{ij}\partial_{ij}\Pi, \label{eq.fluctuation.fokker.planck}
\end{align}
where $\partial_i\equiv\partial/\partial\xi_i,~\boldsymbol{\nu}_j(\boldsymbol{\mu})=\lim_{\Omega\to\infty}f_j(\boldsymbol{n})/\Omega$ and
\begin{equation}
\Gamma_{ij}=\left.\frac{\partial \left[\boldsymbol{V}\boldsymbol{\nu}\right]_i}{\partial \mu_j}\right|_{\boldsymbol{\mu}},\; \boldsymbol{D}=\boldsymbol{V}\;\mathrm{diag}[\boldsymbol{\nu}]\;\boldsymbol{V}^{\top}.
\end{equation}
Multiplying Eq.~\eqref{eq.fluctuation.fokker.planck} with $\xi_i\xi_j$ and integrating by parts gives the following equations of time evolution of variances:
\begin{equation}
\frac{d\boldsymbol{\Sigma}}{dt}=\boldsymbol{\Gamma}\boldsymbol{\Sigma}+\boldsymbol{\Sigma}\boldsymbol{\Gamma}^{\top}+\boldsymbol{D},
\end{equation}
where $\Sigma_{ij}=\langle\xi_i\xi_j\rangle$. We note that the means and variances of molecule numbers can be calculated via $\boldsymbol{\mu},\boldsymbol{\Sigma}$ as $\langle\boldsymbol{n}\rangle=\Omega\boldsymbol{\mu},~\langle (\boldsymbol{n}-\langle \boldsymbol{n}\rangle)(\boldsymbol{n}-\langle \boldsymbol{n}\rangle)^{\top}\rangle=\Omega\boldsymbol{\Sigma}$. At the steady state, the means and variances satisfy the following equations:
\begin{equation}
\boldsymbol{V}\boldsymbol{\nu}=0,\;\boldsymbol{\Gamma}\boldsymbol{\Sigma}+\boldsymbol{\Sigma}\boldsymbol{\Gamma}^{\top}+\boldsymbol{D}=0.
\end{equation}

Because the moment closure yields approximate equations, the solutions of the variety defined by Eq.~\eqref{eq.moments.steady.states} may contain invalid steady states, including negative means or variances. 
Moreover, the variance should be bounded from above by a positive constant or a function of mean to reflect the noise level. 
Consequently, for all $1\leq i\leq N$, we add into the variety the following inequalities: $\mu_i>0,\sigma_{ii}>0$ to ensure that means and variances are positive and $c_i-\sigma_{ii}>0$ or $d_i\mu_i-\sigma_{ii}>0$ to control the noise intensity at the steady state. 
Here, $c_i, d_i$ are positive constants determined by the user. 
For stochastic dynamics, an individual trajectory may undergo large fluctuations while its moment dynamics are time-invariant, which cannot be regarded as a stable steady state from a practical viewpoint. 
These inequality constraints exclude such cases. 
The value $d_i$ can be interpreted as an upper bound of the Fano factor, which is a noise measure that correlates with the distribution width. 
One can also add the inequality $d_i\mu_i^2-\sigma_{ii}>0$ to control the coefficient of variation, which is a measure of the dispersion of a probability distribution. 
Finally, we obtain a semi-algebraic variety where variables are means $\boldsymbol{\mu}$ and variances $\boldsymbol{\sigma}$, and parameters are macroscopic rates of reactions $k_1,k_2,\dots,k_M$. 
The problem is transformed to obtaining conditions for parameters $\boldsymbol{k}$ such that the semi-algebraic variety has exactly one real solution. 
By renaming variables $\boldsymbol{x}=[\boldsymbol{\mu},\boldsymbol{\sigma}]$ and parameters $\boldsymbol{k}=[k_1,k_2,\dots,k_M]$, the semi-algebraic variety can be viewed in the following form:
\begin{equation}\label{eq.semivar.pq}
\begin{cases}
P_1(\boldsymbol{x},\boldsymbol{k})=0,\dots,P_{N_p}(\boldsymbol{x},\boldsymbol{k})=0 \\
Q_1(\boldsymbol{x},\boldsymbol{k})>0,\dots,Q_{N_q}(\boldsymbol{x},\boldsymbol{k})>0
\end{cases}.
\end{equation}
Letting $\mathcal{P}=\{P_1,\dots,P_{N_p}\},\mathcal{Q}=\{Q_1,\dots,Q_{N_q}\}$, we describe semi-algebraic variety of Eq.~\eqref{eq.semivar.pq} as
\begin{equation}\label{eq.semivar}
\mathcal{P}=0,\mathcal{Q}>0.
\end{equation}
If a stochastic reaction network involves $N$ reactant species, then the number of equations is $N_p=N(N+3)/2$, and the number of inequations is $N_q\geq 3N$.

\subsection{Parameter analysis}
We describe an algebraic method to analyze the conditions of parameters such that the semi-algebraic variety defined by Eq.~\eqref{eq.semivar} has exactly one real solution.

\subsubsection{Real solutions in a specific regular system}
If the parameters of the regular system $\{\mathcal{P},\mathcal{Q}\}$ are assigned specific values, then the distinct real solutions of $\{\mathcal{P}=0,\mathcal{Q}>0\}$ can be calculated precisely. 
Suppose that the polynomial set $\mathcal{P}$ has the form $\{P_1,P_2,\dots,P_{N_p}\}$, here $P_i\in\mathbb{R}[x_1,\dots,x_{N_p}],\forall i=1,\dots,N_p$. 
First, we compute distinct real roots $r_1,\dots,r_n$ of polynomial $P_1(x_1)$ and substitute each root $x_1=r_i$ into $\mathcal{P}$ to obtain a new regular set $\mathcal{P}^{(i)}=\{P_2^{(i)},P_3^{(i)},\dots,P_{N_p}^{(i)}\}\subset\mathbb{R}[x_2,\dots,x_{N_p}]$. 
The regular set $\mathcal{P}^{(i)}$ has $N_p-1$ polynomials and if $\{x_2=\alpha_2,\dots,x_{N_p}=\alpha_{N_p}\}$ is a solution of $\mathcal{P}^{(i)}=0$ then $\{x_1=r_i,x_2=\alpha_2,\dots,x_{N_p}=\alpha_{N_p}\}$ is a solution of $\mathcal{P}=0$. 
By repeating this procedure, one can acquire all distinct real solutions of $\mathcal{P}=0$. 
Finally, we only need to examine whether a solution satisfies the conditions $\mathcal{Q}>0$. 
The algorithm $\textbf{ExactSolve}$, which counts the number of distinct real solutions of a regular system, is described in Algorithm~\ref{alg.exactsolve}.

\subsubsection{Regular system decomposition}
A common way to analyze or solve a polynomial system is by computing the triangular decomposition of the system. 
It is known that a semi-algebraic system can be decomposed into several regular systems \cite{Wang.1998,Wang.2000}. 
If we let $\mathcal{P}=\{P_1,\dots,P_{N_p}\},\mathcal{Q}=\{Q_1,\dots,Q_{N_q}\}$ be sets of multivariate polynomials, then one can decompose the set $\mathcal{P}$ into a finite number of regular sets $\mathcal{T}_1,\dots,\mathcal{T}_{J}$ such that
\begin{equation}\label{eq.decomp}
\mathrm{Zero}(\mathcal{P}\setminus\mathcal{Q})=\bigcup_{i=1}^{J}\mathrm{Zero}(\mathcal{T}_i\setminus\mathcal{Q}).
\end{equation}
In this paper, we adopt the algorithm \textbf{RegSer} proposed by Wang \cite{Wang.2000} to decompose the polynomial system. 
To speed up this algorithm, we propose a parallel algorithm that decomposes each polynomial system. 
A sequential algorithm based on \textbf{RegSer} is executed inside this parallel algorithm. 
However, it can be observed that during the decomposition process the algorithm \textbf{RegSer} produces many polynomial systems $\{\mathcal{P},\mathcal{Q}\}$, which has the property $\mathrm{Zero}(\mathcal{P}\setminus\mathcal{Q})=\varnothing$. 
To reduce computation and make the algorithm more efficient, we add a probabilistic test into the sequential decomposition. 
This test probabilistically eliminates polynomial systems that have no complex solution with some specific parameter values. 
The details of each algorithm can be seen in Algorithm~\ref{alg.paradecomp},~\ref{alg.probtest},~\ref{alg.seqdecomp}.
{
\begin{algorithm}[H]
\begin{algorithmic}[1]
\Require A polynomial system $\{\mathcal{P},\mathcal{Q}\}$ in $\mathbb{K}[\boldsymbol{x]}$
\Ensure List of regular sets $\Theta=\left[\mathcal{T}_1,\dots,\mathcal{T}_{J}\right]$
\State set $\Phi \gets [\{\mathcal{P},\mathcal{Q},N_p\}], \Psi \gets \varnothing$
\While {$\Phi \neq \varnothing$}\label{step.beginOfWhile}
\State parallel execute sequential decomposition for each $1\leq i\leq |\Phi|$
\begin{equation*}
[\Phi^{(i)},\Psi^{(i)}]=\mathrm{SequentialDecomposition}(\Phi[i])
\end{equation*}
\State update $\Phi=\bigcup_{i}\Phi^{(i)},~\Psi=\Psi\cup\bigcup_{i}\Psi^{(i)}$
\EndWhile
\State set $\Theta\gets\left[\mathcal{T}_1,\dots,\mathcal{T}_{J}\right]$, here $\Psi=\left[\{\mathcal{T}_1,\mathcal{U}_1\},\dots,\{\mathcal{T}_{J},\mathcal{U}_{J}\}\right]$
\State \Return $\Theta$
\end{algorithmic}
\caption{ParallelDecomposition}\label{alg.paradecomp}
\end{algorithm}
}

To solve the parameter analysis problem, we compute the \textit{border polynomial} $B(\boldsymbol{k})$ \cite{Yang.2001,Yang.2005}, which has the following property: the number of distinct real solutions of the variety defined by Eq.~\eqref{eq.semivar} is invariant over each connected component of the complement of $B(\boldsymbol{k})=0$ in the parameter space. 
First, we define the border polynomial of a regular system as follows.
\begin{definition}
Suppose a regular system $\{\mathcal{P},\mathcal{Q}\}$ is given, then the border polynomial $B_{\{\mathcal{P},\mathcal{Q}\}}(\boldsymbol{k})$ of this system is defined as follows:
\begin{equation*}
B_{\{\mathcal{P},\mathcal{Q}\}}(\boldsymbol{k})\triangleq\prod_{n=1}^{N_p}BP_{n}(\boldsymbol{k})\times\prod_{m=1}^{N_q}BQ_{m}(\boldsymbol{k}),
\end{equation*}
where for all $1\leq n\leq N_p,~1\leq m\leq N_q,$
\begin{align*}
BP_{n}(\boldsymbol{k})&=\mathrm{res}(\mathrm{res}(P_{n},P'_{n},x_{n}),\{P_{1},\dots,P_{n-1}\}),\\
BQ_{m}(\boldsymbol{k})&=\mathrm{res}(Q_{m},\mathcal{P}).
\end{align*}
\end{definition}
For arbitrary semi-algebraic varieties, the border polynomial is defined as follows.
\begin{definition}
Given a semi-algebraic variety as shown in Eq.~\eqref{eq.semivar}, assume that the polynomial set $\mathcal{P}$ is decomposed as Eq.~\eqref{eq.decomp}, then the border polynomial $B(\boldsymbol{k})$ of the variety will be
\begin{equation*}
B(\boldsymbol{k})\triangleq\prod_{i=1}^{J}B_{\{\mathcal{T}_i,\mathcal{Q}\}}(\boldsymbol{k}).
\end{equation*}
\end{definition}
\begin{figure*}[t]
\centering
\includegraphics[width=0.9\textwidth]{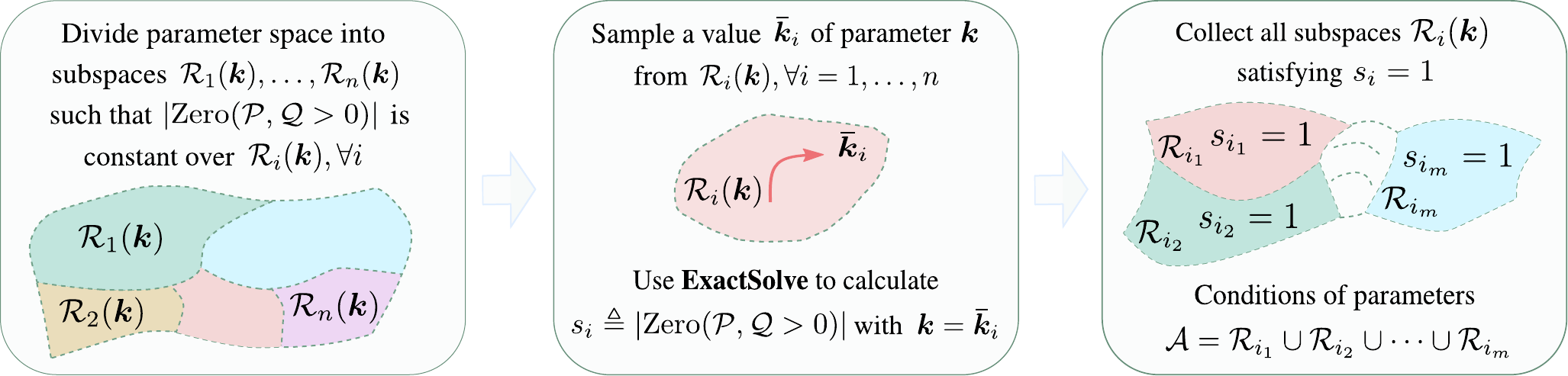}
\caption{(Color online) Outline of parameter analysis in our method. First, as in the left panel, parameter space is divided into a finite number of subspaces such that the number of distinct real solutions of $\{\mathcal{P}=0,\mathcal{Q}>0\}$ is invariant over each subspace. Next, the middle panel shows that a specific value of the parameter is sampled from each subspace to calculate the number of solutions by exploiting the algorithm \textbf{ExactSolve}. Finally, all satisfied subspaces are gathered, and one obtains the conditions of parameters as in the right panel.}\label{fig.parameter-analysis}
\end{figure*}
In the following lemma and theorem, we prove the above-stated property of the border polynomial.
\begin{lemma}\label{lem.regsys}
If a regular system $\{\mathcal{P},\mathcal{Q}\}$ satisfies that $B_{\{\mathcal{P},\mathcal{Q}\}}(\boldsymbol{k})\neq 0$ for all $\boldsymbol{k}\in\mathcal{R}$, where $\mathcal{R}$ is a continuous region, then $\left|\mathrm{Zero}(\mathcal{P},\mathcal{Q}>0)\right|$ is invariant over $\mathcal{R}$.
\end{lemma}
\begin{proof}
As $B_{\{\mathcal{P},\mathcal{Q}\}}(\boldsymbol{k})=\prod_{n=1}^{N_p}BP_{n}(\boldsymbol{k})\times\prod_{m=1}^{N_q}BQ_{m}(\boldsymbol{k})\neq 0$, $BP_{n}(\boldsymbol{k})\neq 0$ and $BQ_{m}(\boldsymbol{k})\neq 0~\forall\boldsymbol{k}\in\mathcal{R}$. 
The conditions $BP_{n}(\boldsymbol{k})\neq 0~(n=1,\dots,N_p)$ imply that the number of distinct real solutions of $\{\mathcal{P}=0\}$ is invariant, while $BQ_{m}(\boldsymbol{k})\neq 0~(m=1,\dots,N_q)$ indicate that each polynomial $Q_{m}$ has no common solution with the system $\{\mathcal{P}=0\}$. 
Let $\boldsymbol{x}_1({\bar{\boldsymbol{k}}}),\dots,\boldsymbol{x}_r({\bar{\boldsymbol{k}}})$ be the distinct real solutions of $\{\mathcal{P}=0\}$ when $\boldsymbol{k}=\bar{\boldsymbol{k}}$. 
It can be seen that $\boldsymbol{x}_i(\bar{\boldsymbol{k}})~(i=1,\dots,r)$ are continuous functions of $\bar{\boldsymbol{k}}$ over region $\mathcal{R}$. 
Assume that indexes $1\leq m\leq N_q,1\leq i\leq r$ and parameters $\boldsymbol{k}_1\neq\boldsymbol{k}_2\in\mathcal{R}$ such that $Q_{m}(\boldsymbol{x}_i({\boldsymbol{k}_1}),\boldsymbol{k}_1)Q_{m}(\boldsymbol{x}_i({\boldsymbol{k}_2}),\boldsymbol{k}_2)<0$ exist. 
Since $Q_{m}(\boldsymbol{x}_i({\boldsymbol{k}}),\boldsymbol{k})$ is a continuous function of $\boldsymbol{k}$, this inequality implies that there exists $\bar{\boldsymbol{k}}\in\mathcal{R}$ such that $Q_{m}(\boldsymbol{x}_i({\bar{\boldsymbol{k}}}),\bar{\boldsymbol{k}})=0$. 
It means that $Q_{m}$ has a common solution with the system $\{\mathcal{P}=0\}$ and is contradictory. 
Therefore, the sign of $Q_{m}(\boldsymbol{x}_i({\boldsymbol{k}}),\boldsymbol{k})$ does not change on region $\mathcal{R}$ for all $m=1,\dots,N_q$ and $i=1,\dots,r$. 
Consequently, $\left|\mathrm{Zero}(\mathcal{P},\mathcal{Q}>0)\right|$ is invariant over region $\mathcal{R}$.
\end{proof}
\begin{theorem}\label{theo.invariant.number.solution}
Let $B(\boldsymbol{k})$ be the border polynomial of the variety defined by Eq.~\eqref{eq.semivar} and $\mathcal{R}$ be a continuous region in the parameter space such that $B(\boldsymbol{k})\neq 0~\forall\boldsymbol{k}\in \mathcal{R}$. 
Then the number of distinct real solutions of the variety is invariant over $\mathcal{R}$.
\end{theorem}
\begin{proof}
Since $B(\boldsymbol{k})=\prod_{i=1}^{J}B_{\{\mathcal{T}_i,\mathcal{Q}\}}(\boldsymbol{k})\neq 0$, it is obvious that $B_{\{\mathcal{T}_i,\mathcal{Q}\}}(\boldsymbol{k})\neq 0~\forall\boldsymbol{k}\in\mathcal{R}$. 
According to Lemma~\ref{lem.regsys}, we obtain the result that $\left|\mathrm{Zero}(\mathcal{T}_i,\mathcal{Q}>0)\right|$ is a constant over $\mathcal{R}$ for all $i=1,\dots,J$. 
Therefore, $\left|\mathrm{Zero}(\mathcal{P},\mathcal{Q}>0)\right|=\sum_{i=1}^{J}\left|\mathrm{Zero}(\mathcal{T}_i,\mathcal{Q}>0)\right|$ is also a constant. 
This means that the number of distinct real solutions of $\{\mathcal{P}=0,\mathcal{Q}>0\}$ is invariant over $\mathcal{R}$.
\end{proof}
After obtaining the border polynomial $B(\boldsymbol{k})$, one can use cylindrical algebraic decomposition \cite{Dennis.1984.1,Dennis.1984.2} to decompose the complement of $B(\boldsymbol{k})=0$ into finitely connected regions such that the sign of $B(\boldsymbol{k})$ does not change over each region. 
The boundaries of these regions are the algebraic surfaces on which $B(\boldsymbol{k})=0$ holds. 
According to Theorem \ref{theo.invariant.number.solution}, the number of distinct real solutions of the variety is invariant over each region. 
By sampling an arbitrary value of the parameters from each region and applying the algorithm \textbf{ExactSolve} to regular systems with specific parameters assigned, we can easily calculate the number of solutions of the variety in each region. 
Ultimately, we gather all satisfied regions and obtain the conditions of the parameters such that the variety has exactly one real solution. 
The schematic of our method is shown in Fig.~\ref{fig.parameter-analysis}.

Given all of these results, we propose the following algorithm to solve the problem of parameter analysis.
{
\begin{algorithm}[H]
\caption{ParameterAnalysis}\label{alg.parana}
\begin{algorithmic}[1]
\Require A polynomial system $\{\mathcal{P}=0,\mathcal{Q}>0\}$
\Ensure The conditions of parameters $\boldsymbol{k}$ such that the system has exactly one real solution.
\State set initial conditions of parameters $\mathcal{A}\gets\varnothing$
\State use algorithm \textbf{ParallelDecomposition} to decompose polynomial set $\mathcal{P}$ as Eq.~\eqref{eq.decomp}
\State compute the border polynomial $B(\boldsymbol{k})$ of the system $\{\mathcal{P},\mathcal{Q}\}$
\State decompose the complement of $B(\boldsymbol{k})=0$ into finitely connected cells $\mathcal{R}_1,\dots,\mathcal{R}_n$
\For{$i=1,\dots,n$}
\State sample an arbitrary point $\bar{\boldsymbol{k}}_i$ from cell $\mathcal{R}_i$
\State \begin{varwidth}[t]{0.87\linewidth}
substitute $\boldsymbol{k}=\bar{\boldsymbol{k}}_i$ into the system $\{\mathcal{P}=0,\mathcal{Q}>0\}$ and use algorithm \textbf{ExactSolve} to count the number of distinct real solutions $s_i\triangleq\left|\mathrm{Zero}(\mathcal{P},\mathcal{Q}>0)\right|$ 
\end{varwidth}
\State set $\mathcal{A}\gets\mathcal{A}\cup\mathcal{R}_i$ if $s_i=1$
\EndFor
\State \Return $\mathcal{A}$
\end{algorithmic}
\end{algorithm}
}

\section{Results}\label{sect.res}
Here, we illustrate how our proposed method works in typical stochastic reaction networks. 
For each case, we compute the results with the approximations based on normal distributions and the LNA individually.
To verify the validity of our proposed method, stochastic simulations are executed. 
For each specific parameter value, we run $10^6$ realizations of stochastic trajectories to obtain the means and variances at the steady state. 
First, we perform numerical simulations to find the boundaries of the regions of satisfied parameters. 
The boundary here is considered to be the place at which a given constraint condition becomes broken. 
After that, we uniformly sample many points from the parameter space to determine the interior of satisfied regions. 
By sequentially applying these procedures, we obtain the simulation results for each considered case.

\subsection{Gene regulatory system}
We consider a simple single gene regulatory system \cite{Thattai.2001,Hortsch.2018}. 
The system contains four reactions describing the transcription, translation, and degradation of the mRNA and protein as follows (Fig.~\ref{fig.gene}(a)):
\begin{equation}
\varnothing\mathrel{\mathop{\rightleftarrows}_{k_{1}}^{k_0}}M,M\xrightarrow{k_2}M+P,P\xrightarrow{k_3}\varnothing.
\end{equation}
Here we denote the mRNA and protein by $M$ and $P$, respectively, and
$k_i~(0\leq i\leq 3)$ denote the reaction rates.
The stoichiometric matrix $\boldsymbol{V}$ of the system is
\begin{equation*}
\setlength\arraycolsep{7pt}
\boldsymbol{V}=\begin{bmatrix}
1 & -1 & 0 & 0\\
0 & 0 & 1 & -1
\end{bmatrix}.
\end{equation*}
Let $n_M(t)$ and $n_P(t)$ be the copy numbers of $M$ and $P$, respectively, at time $t$. 
Assuming mass-action kinetics, the propensity functions are given by
\begin{equation*}
\boldsymbol{f}(n_M,n_P)=\left[\Omega k_0,k_1n_M,k_2n_M,k_3n_P\right]^{\top}.
\end{equation*}
 The corresponding CME is as follows:
\begin{equation}
\begin{aligned}
\frac{\partial P(n_M,n_P,t)}{\partial t}&=\Omega k_0P(n_M-1,n_P,t)\\
&\hspace{-1.8cm}+k_1(n_M+1)P(n_M+1,n_P,t)\\
&\hspace{-1.8cm}+k_2n_MP(n_M,n_P-1,t)\\
&\hspace{-1.8cm}+k_3(n_P+1)P(n_M,n_P+1,t)\\
&\hspace{-1.8cm}-(\Omega k_0+k_1n_M+k_2n_M+k_3n_P)P(n_M,n_P,t).
\end{aligned}
\end{equation}
It has been known that depending on the parameters $k_i~(0\leq i\leq 3)$, this system may exhibit large fluctuations while the means are invariant.
The Fano factors are often used to characterize the noise strength at the steady state.
The analytical results \cite{Thattai.2001} have revealed that the dynamics of the mRNA is a Poisson process, i.e., the Fano factor is 1, while the dynamic process of protein is super-Poissonian (the Fano factor is larger than 1).
Therefore, we put here emphasis on inequation constraints of the Fano factor of the protein.

We denote by $\mu_{P}$ and $\sigma_{P}$ the mean and variance, respectively, of the protein at the steady state.
We add a constraint $\sigma_P<2\mu_P$ to set an upper bound of the Fano factor of $P$.
We fix parameters $\Omega=100,k_0=1$ and analyze three cases: the conditions of parameters $(k_1,k_2)$ when $k_3=10$, of $(k_2,k_3)$ when $k_1=10$, and of $(k_1,k_3)$ when $k_2=10$ such that the steady-state distribution satisfies above constraint.
The analytical and numerical results of these three cases are shown in Figs.~\ref{fig.gene}(b)-(d), respectively, in which colored regions represent conditions obtained by the proposed method.
Specifically, the blue lines and violet dashed lines express the boundaries of the regions obtained with approximations based on normal distribution and the LNA, respectively.
The circles denote those with numerical simulations.
\begin{figure}[t]
\centering
\includegraphics[width=0.23\textwidth]{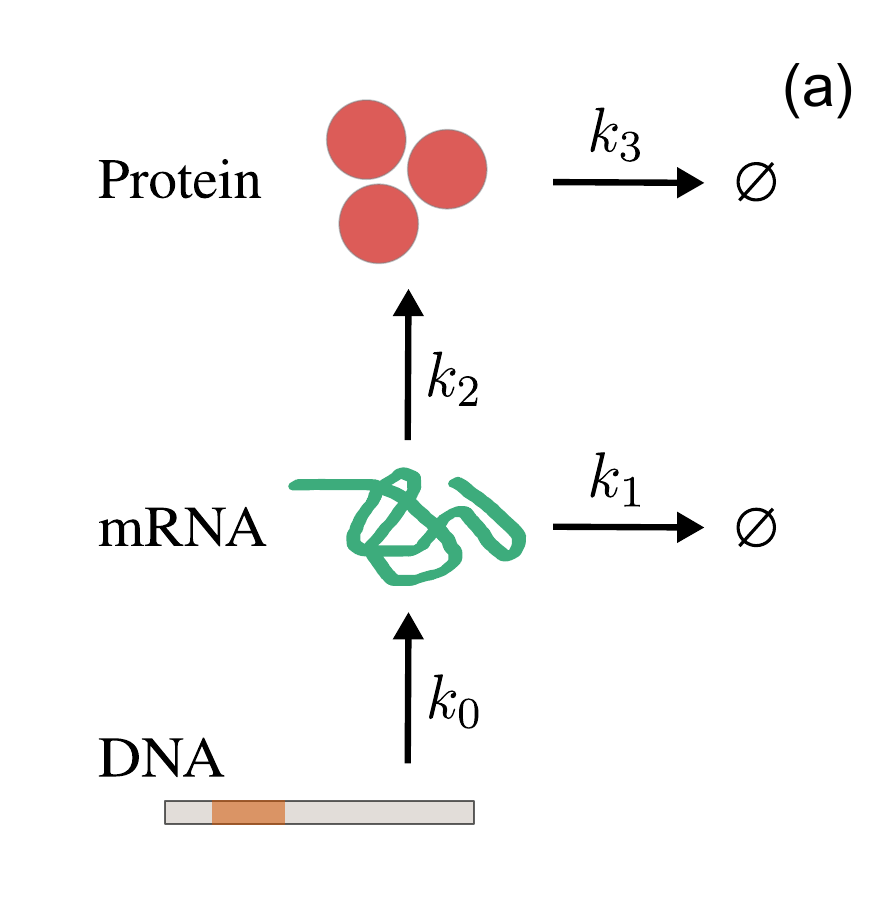}
\includegraphics[width=0.23\textwidth]{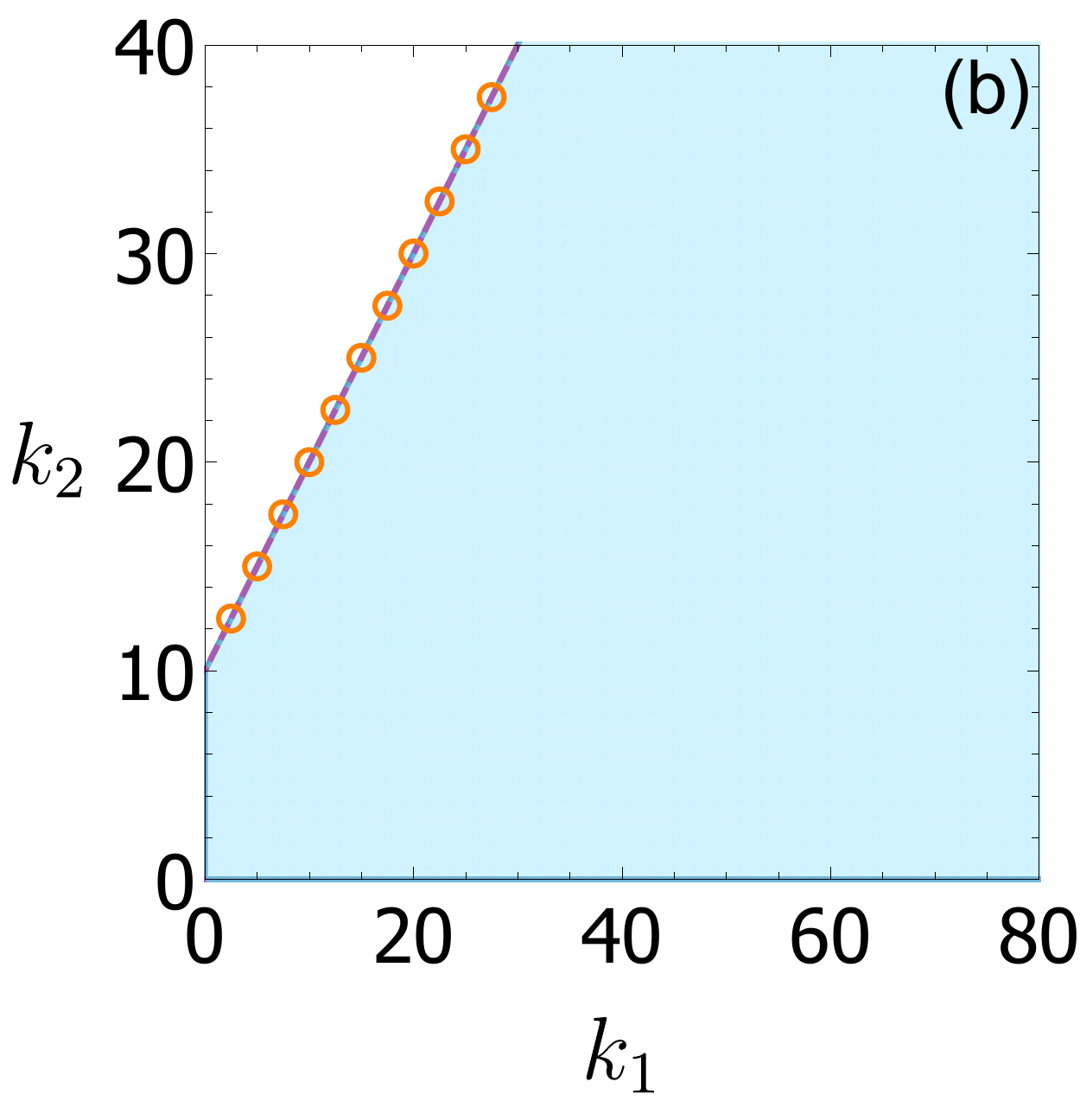}
\includegraphics[width=0.23\textwidth]{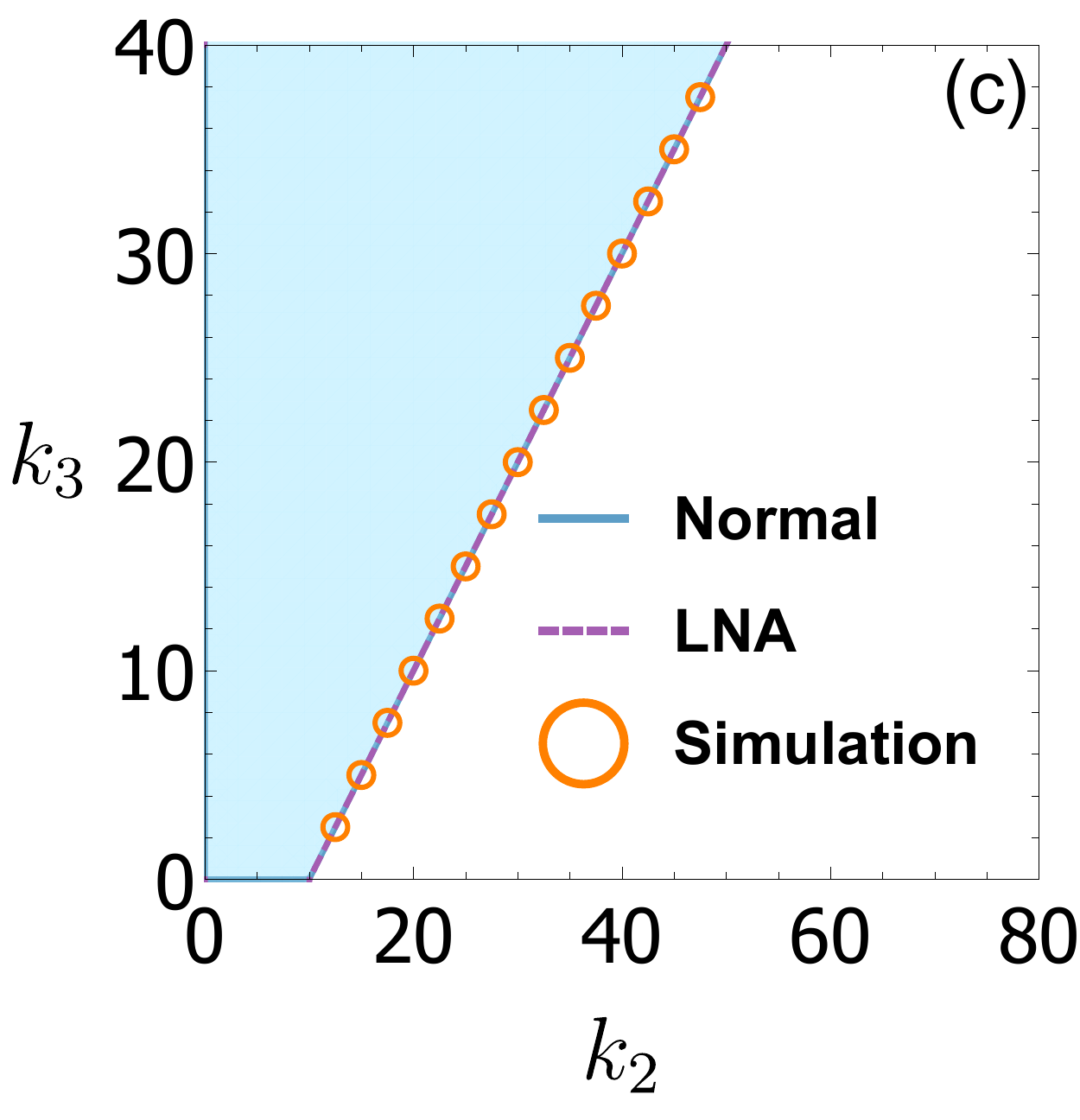}
\includegraphics[width=0.23\textwidth]{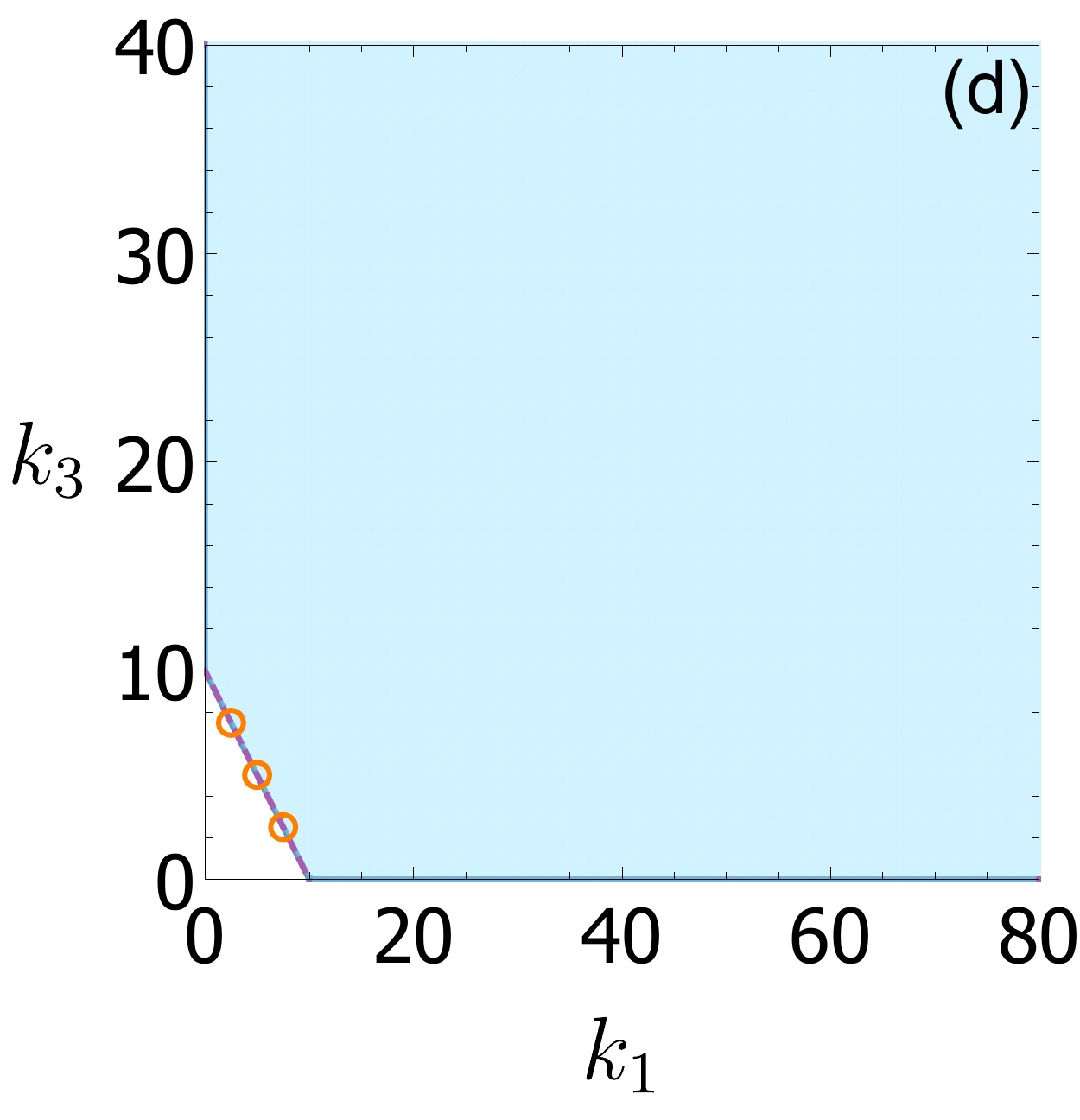}
\caption{(Color online) Parameter analysis of gene regularoty system with fixed parameters $\Omega=100,k_0=1$. Our analytical results (blue regions) represent the conditions of parameters such that the inequality $\sigma_P<2\mu_P$ is satisfied at the steady state. Blue lines express the boundaries of regions obtained with the approximation based on normal distribution. Violet dashed lines represent the counterpart obtained with the LNA. The simulation results (orange circles) indicate the boundaries of regions of satisfied parameters. Figures from left to right correspond to (a) schematic diagram of a simple gene regulatory system, conditions of parameters when (b) $k_3=10$, (c) $k_1=10$, and (d) $k_2=10$.}\label{fig.gene}
\end{figure}

It can be seen that our results agree with simulation results. 
Since the system is linear, the moment equations obtained with the approximation based on normal distribution and the LNA are same, and these equations are not approximate but exact ones.
Therefore, the results obtained with both of approximation schemes are identical with simulation.

\subsection{Michaelis--Menten enzyme reactions}
The Michaelis--Menten enzyme reactions can be described by
\begin{equation}
\varnothing\xrightarrow{k_0}S,\quad E+S\mathrel{\mathop{\rightleftarrows}_{k_{2}}^{k_1}} ES\xrightarrow{k_3} E+P,
\end{equation}
where $E,S,ES$, and $P$ represent the free enzyme, input substrate, enzyme-substrate complex, and product, respectively, and $k_i~(0\leq i\leq 3)$ denotes the reaction rates. 
Let $n_E(t),n_S(t),n_{ES}(t)$, and $n_P(t)$ be the molecule numbers of reactant species $E, S, ES$, and $P$, respectively, at time $t$. 
This reaction network has a conservation relation $n_E(t)+n_{ES}(t)=n_T$, where $n_T$ is a constant positive integer. 
This implies that the sum of the molecule numbers of $E$ and $ES$ is constant at all times. 
Therefore, the behavior of this system can be characterized by two variables $n_E(t)$ and $n_S(t)$. 
The stoichiometric matrix $\boldsymbol{V}$ is
\begin{equation*}
\setlength\arraycolsep{7pt}
\boldsymbol{V}=\begin{bmatrix}
0 & -1 & 1 & 1\\
1 & -1 & 1 & 0
\end{bmatrix}.
\end{equation*}
Assuming mass-action kinetics, the propensity functions are given by
\begin{equation}\label{eq.enzyme.propensity}
\begin{aligned}
\boldsymbol{f}(n_E,n_S)&=\\
&\hspace{-1cm}\left[\Omega k_0,~\Omega^{-1}k_1n_En_S,~k_2(n_T-n_E),~k_3(n_T-n_E)\right]^{\top}.
\end{aligned}
\end{equation}
Substituting propensity functions in Eq.~\eqref{eq.enzyme.propensity} into Eq.~\eqref{eq.cme}, we obtain the corresponding CME as follows:
\begin{equation}
\begin{aligned}
\frac{\partial P(n_E,n_S,t)}{\partial t}&=\Omega k_0P(n_E,n_S-1,t)\\
&\hspace{-1.5cm}+\Omega^{-1}k_1(n_E+1)(n_S+1)P(n_E+1,n_S+1,t)\\
&\hspace{-1.5cm}+(n_T-n_E+1)\Big(k_2P(n_E-1,n_S-1,t)\\
&\hspace{-1.5cm}+k_3P(n_E-1,n_S,t)\Big)-\Big(\Omega k_0+\displaystyle\Omega^{-1}k_1n_En_S\\
&\hspace{-1.5cm}+(k_2+k_3)(n_T-n_E)\Big)P(n_E,n_S,t).
\end{aligned}
\end{equation}
We fix parameters $k_2=5$ and $k_3=4$ and analyze the conditions of parameters $k_0$ and $k_1$. 
Let $\mu_{E}$ and $\sigma_{E}$ be the mean and variance, respectively, of the species $E$ at the steady state ($\mu_{S}$ and $\sigma_{S}$ are defined analogously for $S$). 
First, we add two constraints: $\sigma_{E}<c$ and $\sigma_{S}<c$, where $c$ is a positive constant. 
The conditions of parameters such that the means and variances satisfy the above constraints are shown in Figs.~\ref{fig.m2enzyme}(a)--(c), where we show three $c$ cases: (a) $c=200$, (b) $c=300$, and (c) $c=400$. 
Next, we fix $c=500$ and add two additional constraints: $\sigma_{E}<d\mu_{E}$ and $\sigma_{S}<d\mu_{S}$, where $d$ is a positive constant. 
The value of $d$ can be considered as an upper bound of the Fano factor, which is equal to one in the case of the Poisson distribution. 
Here, we set several intermediate ranges for the Fano factor and accept, to some extent, a large dispersion in the steady-state distribution. 
Figures~\ref{fig.m2enzyme}(d)--(f) show results for three $d$ cases with fixed $c=500$: (d) $d=1.5$, (e) $d=2.0$, and (f) $d=2.5$. 
In Fig.~\ref{fig.m2enzyme}, the meanings of the colored regions, lines, and circles are the same as in Fig.~\ref{fig.gene}, and the other parameter values are shown in the caption of Fig.~\ref{fig.m2enzyme}. 
\begin{figure}[t]
\centering
\includegraphics[width=0.23\textwidth]{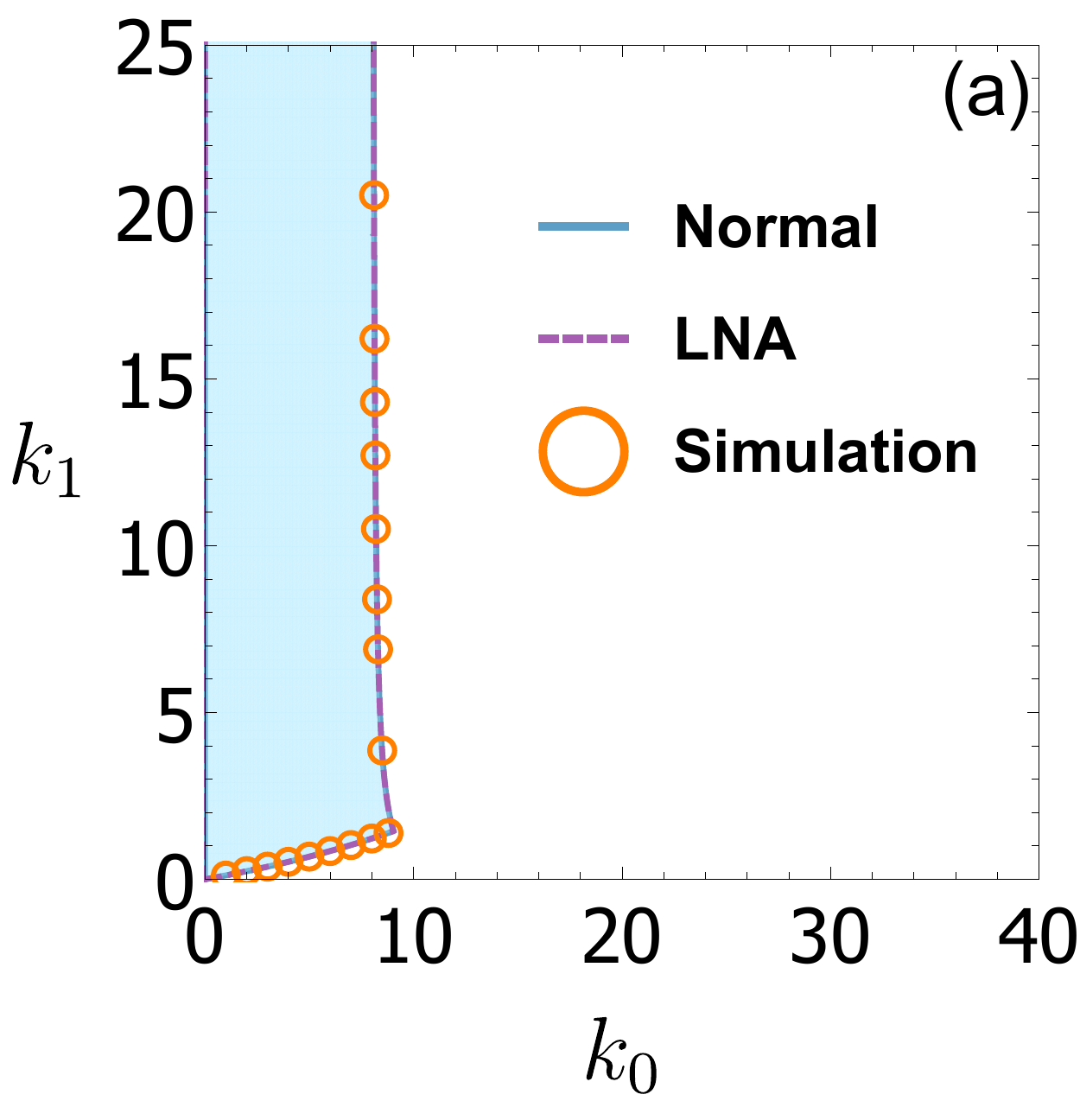}
\includegraphics[width=0.23\textwidth]{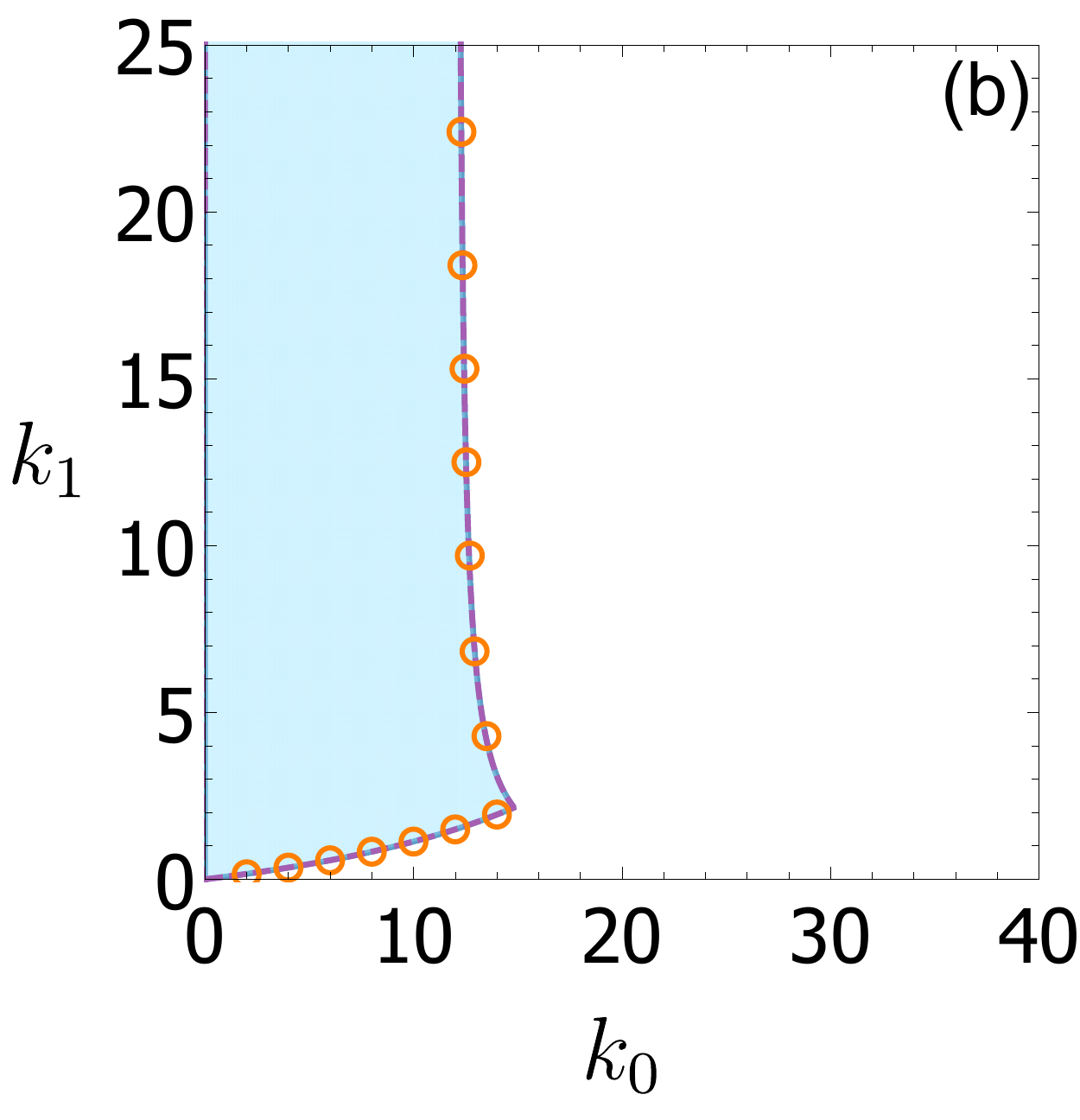}
\includegraphics[width=0.23\textwidth]{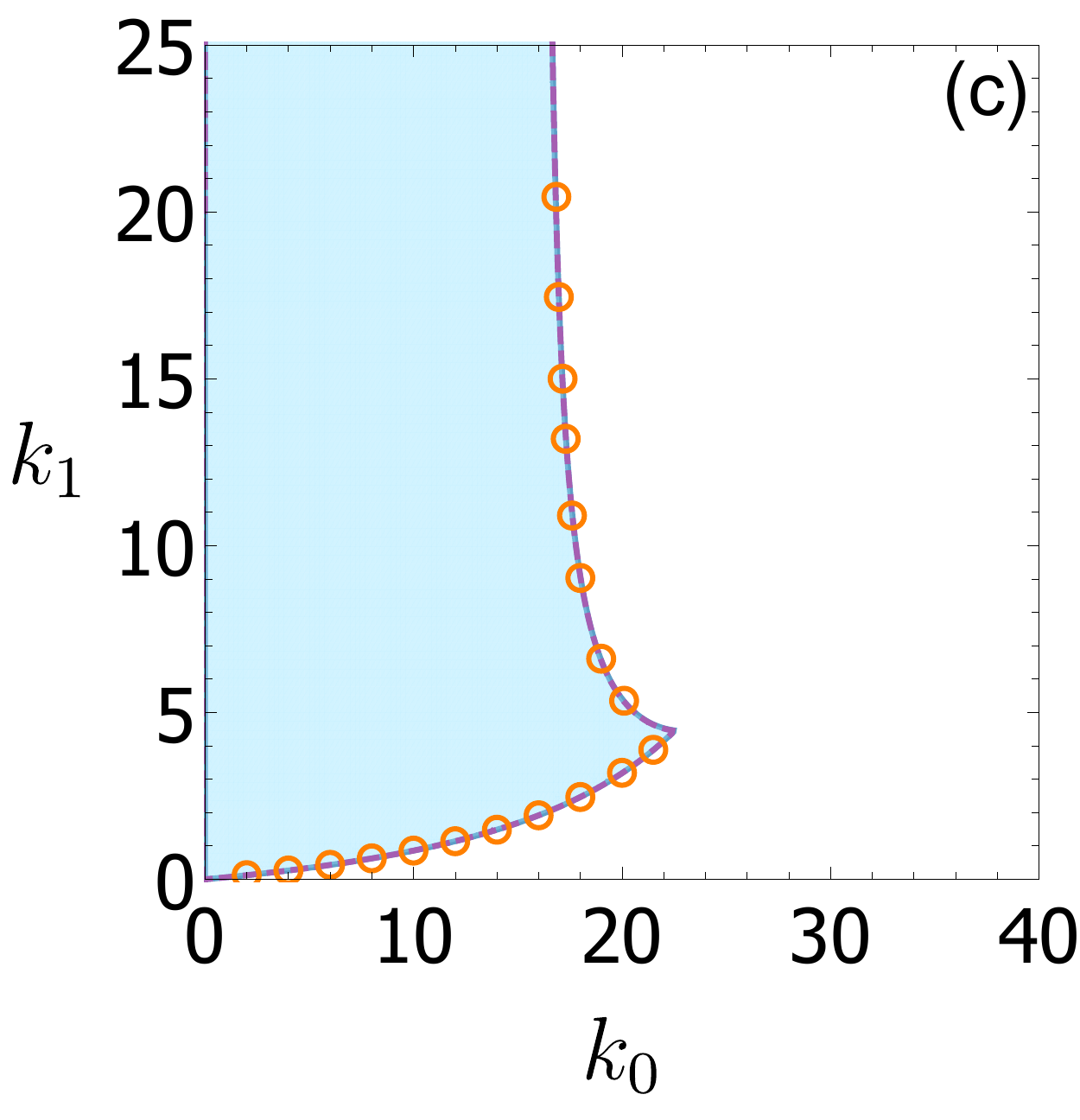}
\includegraphics[width=0.23\textwidth]{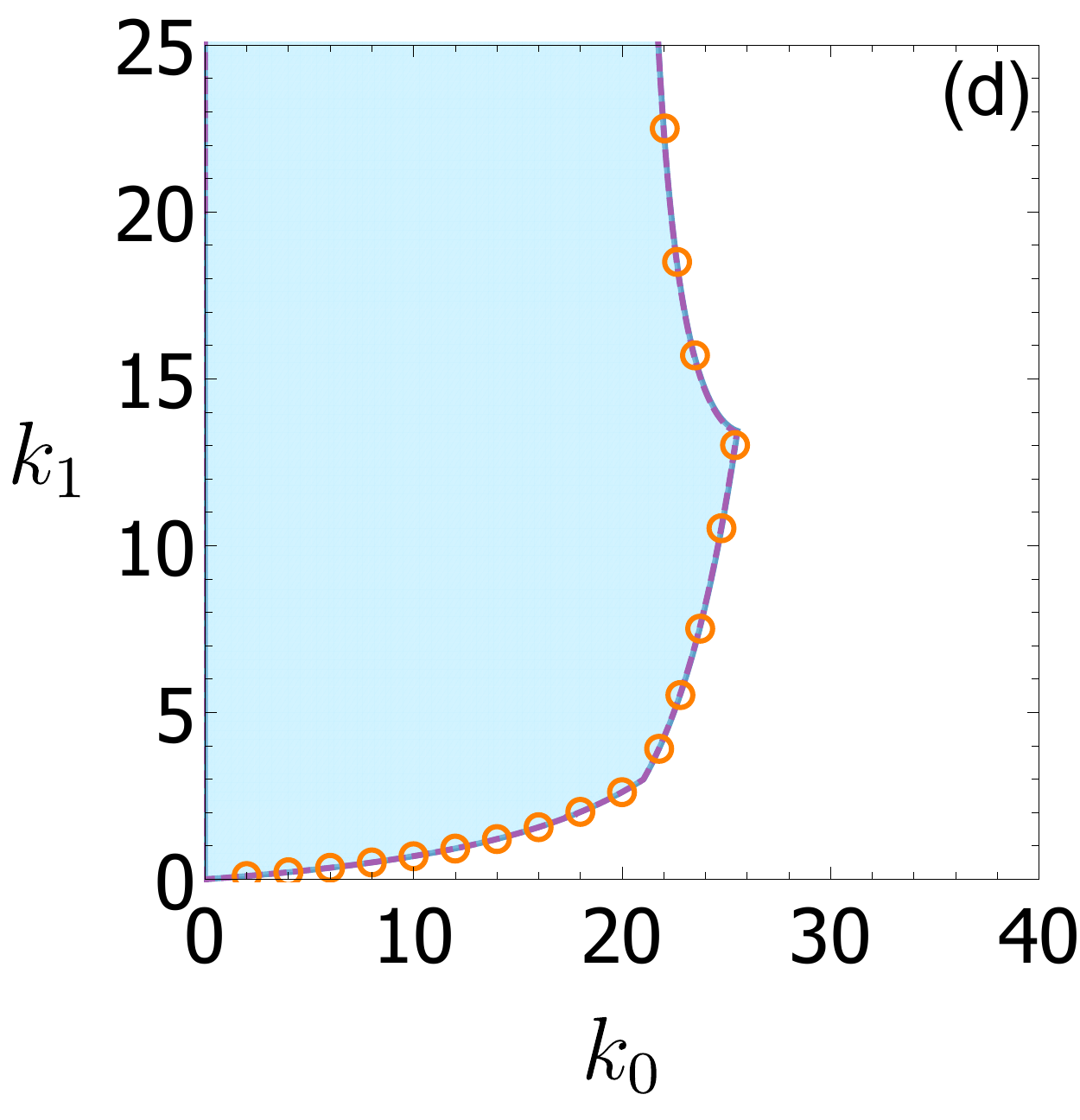}
\includegraphics[width=0.23\textwidth]{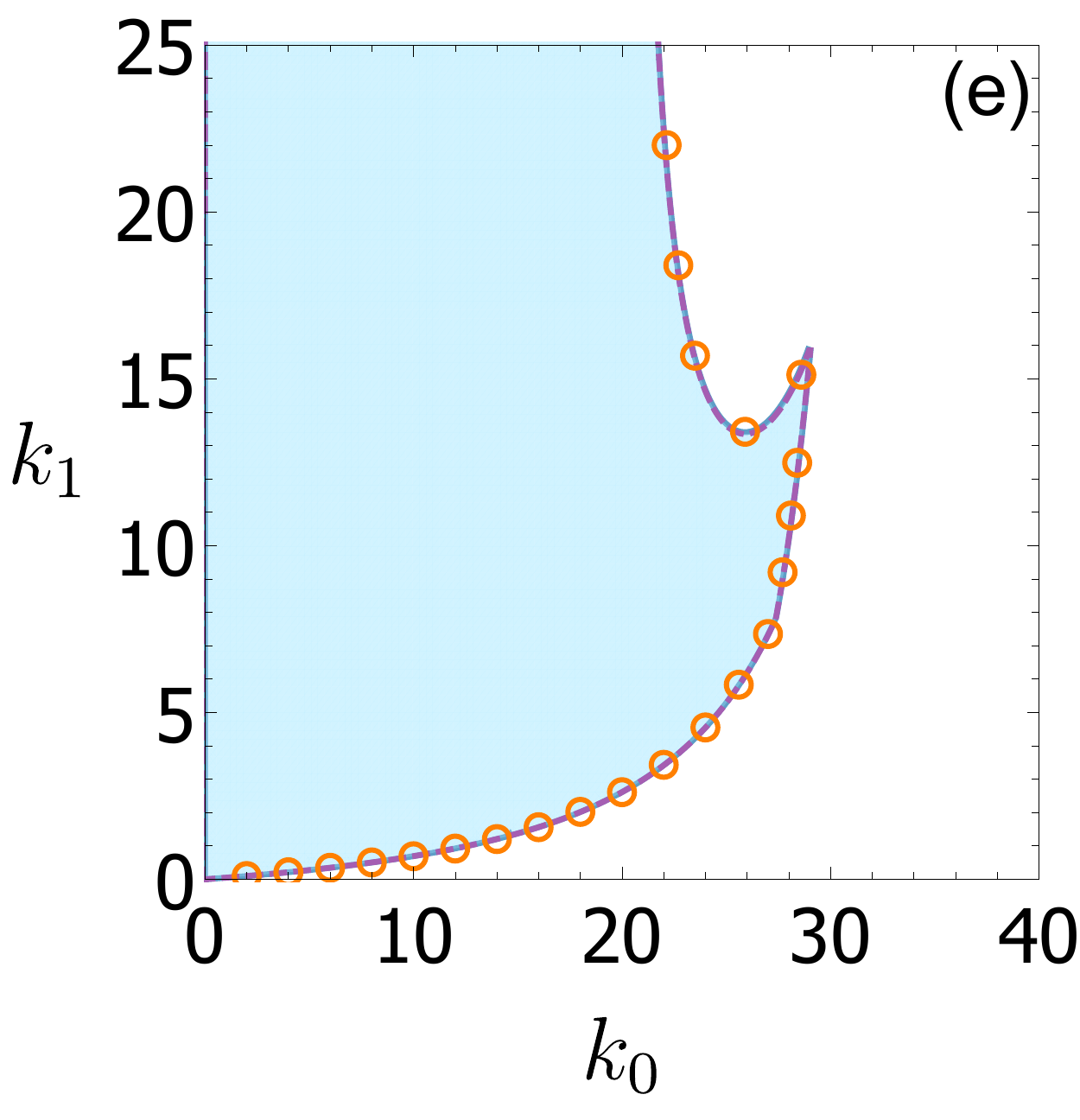}
\includegraphics[width=0.23\textwidth]{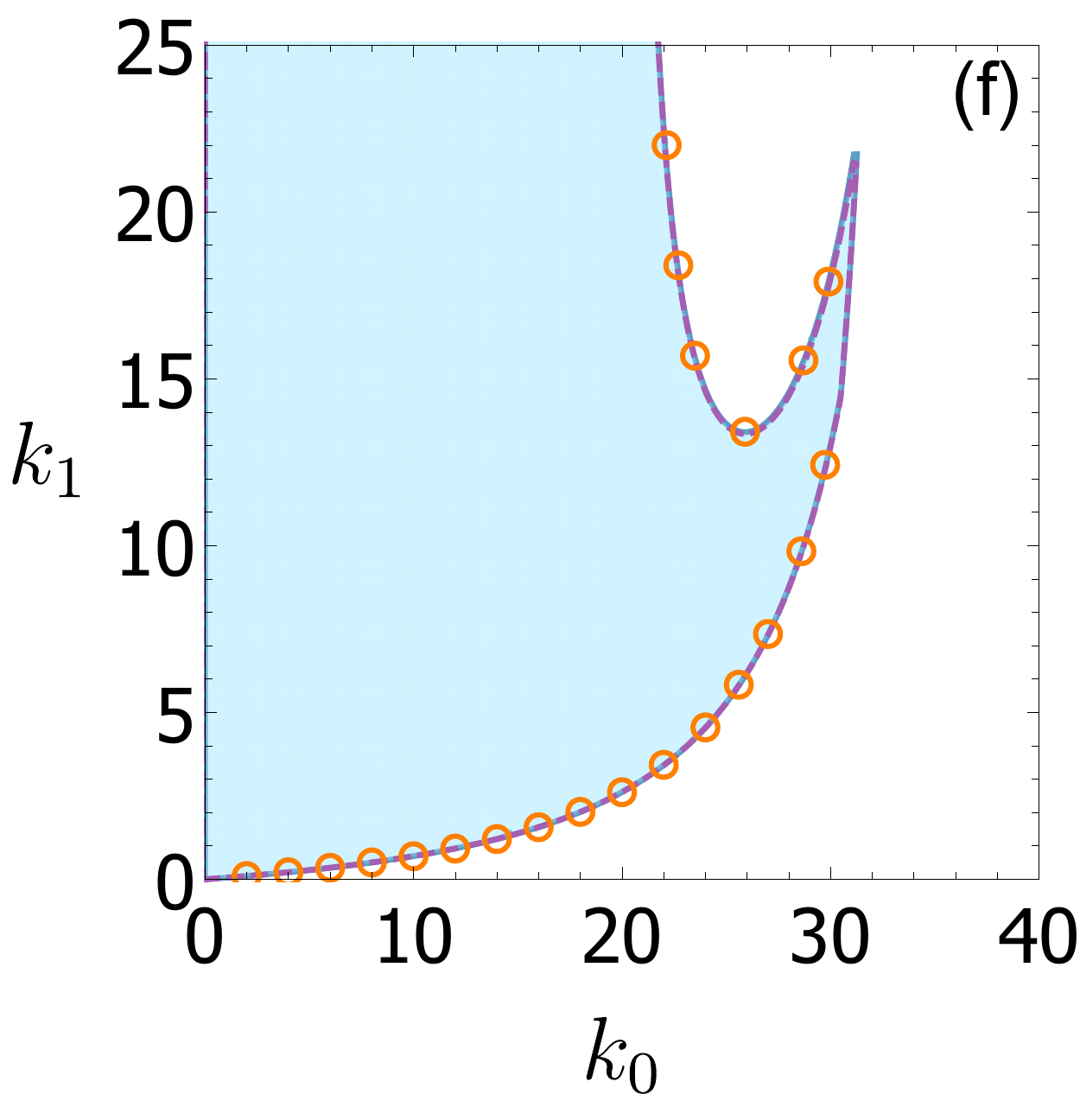}
\caption{(Color online) Parameter analysis of enzyme reactions with fixed parameters $\Omega=100, n_T=1000, k_2=5$, and $k_3=4$. Our analytical results (blue region) represent the region of parameters such that inequalities $\sigma_{E}<c,\sigma_{S}<c$ are satisfied at the steady state, where (a) $c=200$, (b) $c=300$, and (c) $c=400$. The boundaries of regions obtained with the approximations based on normal distribution and the LNA are expressed by blue lines and violet dashed lines, respectively. Orange circles, which represent simulation results, indicate the boundaries of regions of satisfied parameters. When $c$ is fixed to $500$ and constraints $\sigma_{E}<d\mu_{E}$ and $\sigma_{S}<d\mu_{S}$ are added, the corresponding results are as in (d) $d=1.5$, (e) $d=2.0$, and (f) $d=2.5$.}\label{fig.m2enzyme}
\end{figure}
When we relax the constraints, the region of parameters also enlarges. 
From the figures, for Michaelis--Menten enzyme reactions that often appear in biochemical reactions, we can conclude that our method gives results consistent with stochastic simulations. 
Blue lines and violet dashed lines are almost identical.
This implies that the approximations based on normal distribution and the LNA give the same accuracy in these cases.
Interestingly, when we relax the constraint conditions $\sigma_{E}<d\mu_{E}$ and $\sigma_{S}<d\mu_{S}$, the region of the parameters undergoes a change to non-convex as shown in Figs.~\ref{fig.m2enzyme}(e) and (f). 
The upper boundary curve indicates where the condition $\sigma_{S}<c$ becomes violated, i.e., when $\sigma_{S}$ is equal to $c$. 
This yields the discontinuity of the parameter $k_0$ if $k_1$ is fixed near $k_1=15$.

\subsection{Brusselator model}
Next, we examine a nonlinear oscillating reaction network, the Brusselator model \cite{Prigogine.1968}. 
This system is composed of two reactant species $X_1$ and $X_2$ and the following four reactions:
\begin{equation}
\varnothing\xrightarrow{k_0} X_1,~2X_1+X_2\xrightarrow{k_1} 3X_1,~X_1\xrightarrow{k_2} X_2,~X_1\xrightarrow{k_3}\varnothing,
\end{equation}
where $k_i~(0\leq i\leq 3)$ denote the reaction rates. 
We fix $k_0=k_3=1$ and consider $k_1$ and $k_2$ parameters as target of the analysis. Let $n_1(t)$ and $n_2(t)$ be the molecule numbers of reactant species $X_1$ and $X_2$, respectively, at time $t$. 
The stoichiometric matrix $\boldsymbol{V}$ is
\begin{equation*}
\setlength\arraycolsep{7pt}
\boldsymbol{V}=\begin{bmatrix}
1 & 1 & -1 & -1\\
0 & -1 & 1 & 0
\end{bmatrix}.
\end{equation*}
Assuming mass-action kinetics, the propensity functions are given by
\begin{equation}
\boldsymbol{f}(n_1,n_2)=\left[\Omega k_0,~\Omega^{-2}k_1n_1(n_1-1)n_2,~k_2n_1,~k_3n_1\right]^{\top}.
\end{equation}
The master equation of the system is as follows:
\begin{equation}
\begin{aligned}
\frac{\partial P(n_1,n_2,t)}{\partial t}&=\Omega k_0P(n_1-1,n_2,t)\\
&\hspace{-1.8cm}+\Omega^{-2}k_1(n_1-1)(n_1-2)(n_2+1)P(n_1-1,n_2+1,t)\\
&\hspace{-1.8cm}+(n_1+1)\Big(k_2P(n_1+1,n_2-1,t)+k_3P(n_1+1,n_2,t)\Big)\\
&\hspace{-1.8cm}-\Big(\Omega k_0+\Omega^{-2}k_1n_1(n_1-1)n_2+(k_2+k_3)n_1\Big)\\
&\hspace{3cm}\times P(n_1,n_2,t).
\end{aligned}
\end{equation}
In the case of the deterministic model, depending on the magnitude relation of $k_2$ and $k_1+1$, the deterministic rate equations show sustained oscillations, damped oscillations, or overdamped oscillations. 
However, stochastic and deterministic models may behave qualitatively differently for some parameters. 
For instance, a stochastic model can exhibit a sustained oscillation where its corresponding deterministic model shows an overdamped oscillation. 
Let $\mu_{X_1}$ and $\sigma_{X_1}$ be the mean and variance, respectively, of the species $X_1$ at the steady state ($\mu_{X_2}$ and $\sigma_{X_2}$ are defined analogously for $X_2$). 
We add two constraints, $\sigma_{X_1}<c$ and $\sigma_{X_2}<c$, where $c$ is a positive constant, to control the noise level or amplitude of the oscillation at the steady state. 
We calculate the conditions of the parameters such that the steady-state distribution satisfies these constraints. 
The results are shown in Fig.~\ref{fig.brusselator} for three $c$ cases: (a) $c=300$, (b) $c=400$, and (c) $c=500$. 
Again, meanings of the colored regions, lines, and circles in Fig.~\ref{fig.brusselator} are the same as in Fig.~\ref{fig.gene}, and the other parameter values are shown in the caption of Fig.~\ref{fig.brusselator}.
\begin{figure}[t]
\centering
\includegraphics[width=0.23\textwidth]{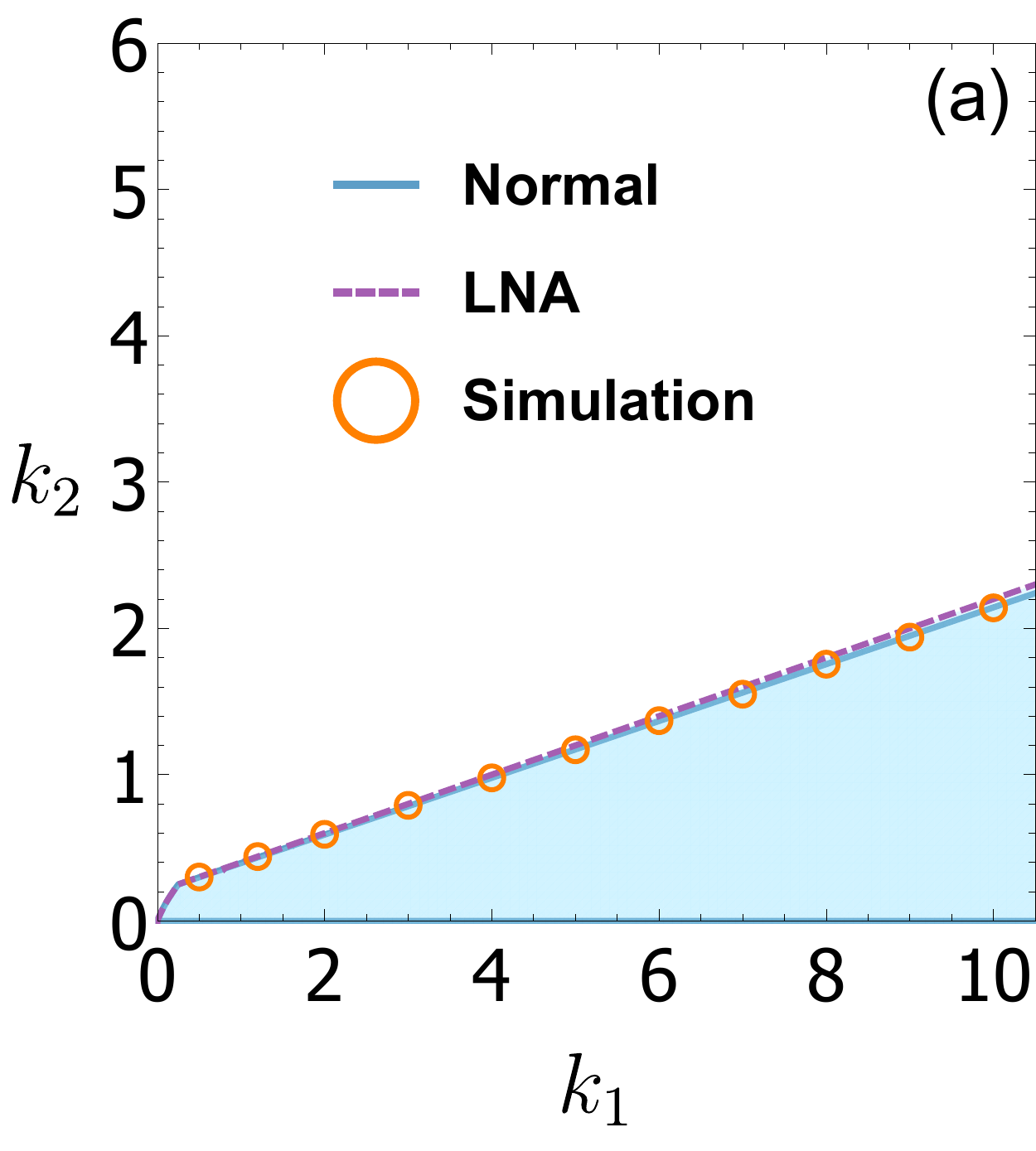}
\includegraphics[width=0.23\textwidth]{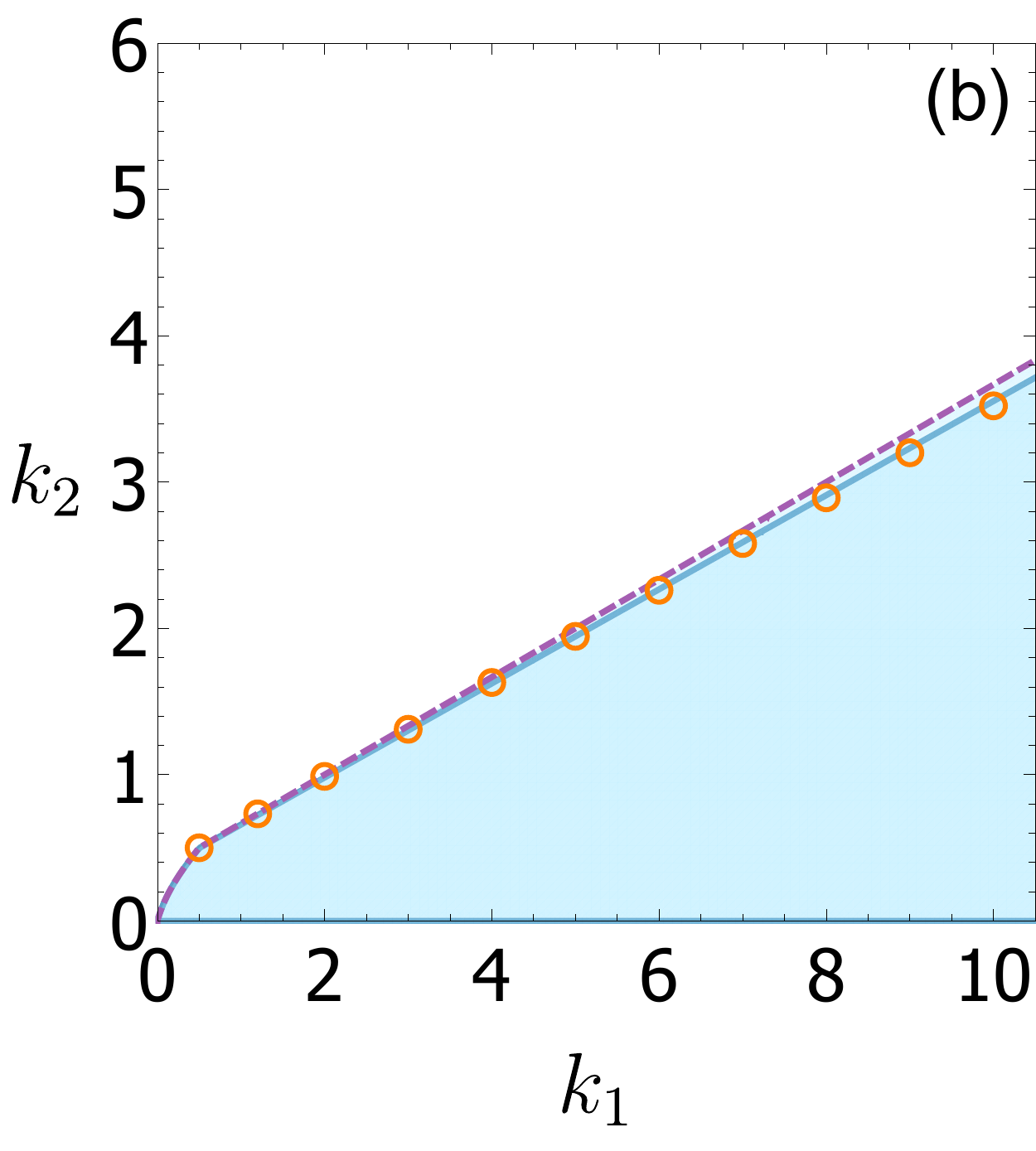}
\includegraphics[width=0.23\textwidth]{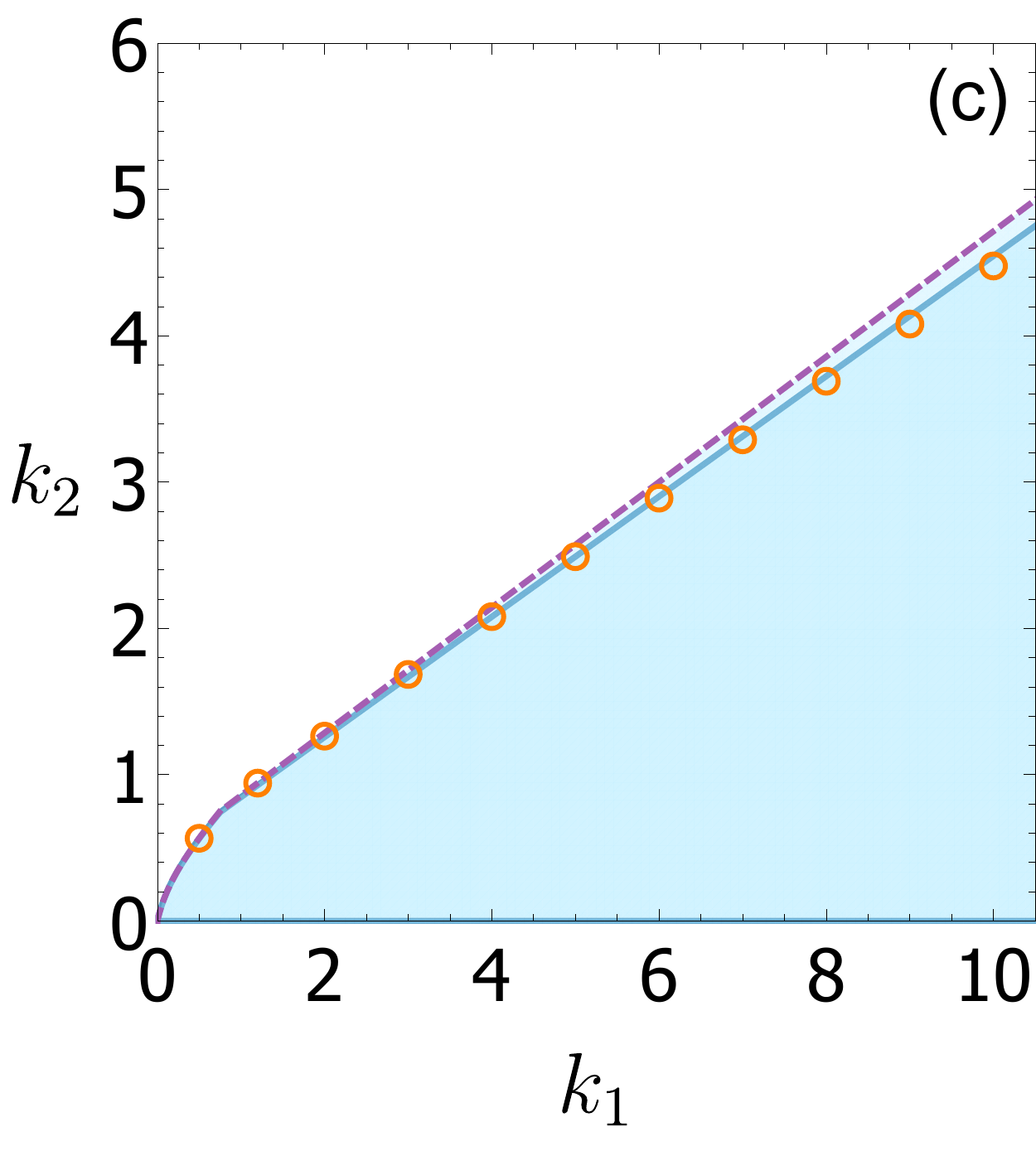}
\caption{(Color online) Parameter analysis of the Brusselator model with fixed parameters $\Omega=200$ and $k_0=k_3=1$. Blue region represents the conditions of parameters such that inequalities $\sigma_{X_1}<c,\sigma_{X_2}<c$ are satisfied at the steady state. Blue lines and violet dashed lines express the boundaries of regions obtained with the approximations based on normal distribution and the LNA, respectively. The simulation results (orange circles) indicate the boundaries of regions of satisfied parameters. Figures from left to right correspond to the cases of (a) $c=300$, (b) $c=400$, and (c) $c=500$.}\label{fig.brusselator}
\end{figure}
When $k_1$ is fixed, increasing $k_2$ results in stronger fluctuations at the steady state. 
From Fig.~\ref{fig.brusselator}, it can be seen that our method gives results comparable with those of stochastic simulations.
When increasing the value of $c$, however, the results obtained with the LNA are not as good as with the approximation based on normal distribution.
The reason is that the LNA is derived under the assumption of small fluctuations, which are of the order of $\mathcal{O}(\Omega^{1/2})$.
Increasing $c$ means that we allow more considerable fluctuations at the steady state, and leading to poor performance of the LNA.
The LNA underestimates the values of variances at the steady state; therefore, the parameter regions are enlarged.

In the case of the deterministic model, the system possesses a stable steady state when $k_2<k_1+1$. 
However, in the case of the stochastic model, the region of parameters is much smaller.
This can be explained by the fact that in the stochastic model, noise-induced oscillation occurs earlier, which means that the region of parameters is limited.

\section{Conclusions}\label{sect.conc}
Stochastic fluctuations are inevitable and ubiquitous in biological systems. 
Recent experimental studies have revealed that noise plays a crucial role in the biochemical reaction networks of living cells. 
For example, stochastic effects on gene expression lead to massive amounts of cell--cell variation observed in isogenic populations \cite{Elowitz.2002}. 
Deterministic models, i.e., rate equations, give a macroscopic description of the dynamics of reaction networks and are incapable of capturing the features of the system when the effects of stochastic fluctuations become significant. 
Thus, stochastic models, i.e., master equations, are typically exploited to describe the dynamics of reaction networks as a stochastic process.

Steady states play several important roles in many biological functions and have been intensively studied in recent years. 
However, stochastic models have not been explored. 
There is a possibility that certain steady state in the stochastic model cannot be observed with the deterministic model \cite{Patrick.2015}. 
In this paper, we proposed an algebraic method to calculate parameter regions in which steady-state distribution of the reaction network satisfies some given constraints. 
We examined our method on three small reaction networks and performed numerical simulations to verify its validity. 
Through the experiments, it can be concluded that our method gives consistent results with those of the simulations. 
Our approach does not require prior knowledge of the parameters. 
This is a significant benefit since information about the parameters is often unavailable in biological systems. 
One can also intentionally add constraints, which relate means and variances at steady state, to obtain the conditions of parameters under which the system is brought into a desired steady state.

The precision of our method relies on moment closure approximation. 
Approximations in our method (moment closure based on normal distributions and LNA) give comparable results in the cases of unimodal steady-state distributions. 
For the networks characterized by a multimodal distribution, these approximations may provide unreliable results \cite{Ramon.2012,David.2014,David.2015}.
In such cases, an approximation scheme that can handle multimodal distributions like conditional LNA \cite{Thomas.2014} should be considered.
We stress that the approximation of moment equations in our method can be flexibly replaced by other approximation schemes.

\section*{Acknowledgments}
This work was supported by MEXT KAKENHI Grant No.JP16K00325.

\begin{widetext}
\appendix

\section{Derivation of moment equations}\label{app.moment.equations}
To obtain the equation of the first moment, we multiply $n_i$ by Eq.~\eqref{eq.cme} and take the sum of all possible states $\boldsymbol{n}$ to get the following equation:
\begin{equation}\label{eq.mean1}
\sum_{\boldsymbol{n}}n_i\frac{\partial P(\boldsymbol{n},t)}{\partial t}=\sum_{j=1}^{M}\sum_{\boldsymbol{n}}(n_{i}f_{j}(\boldsymbol{n}-\boldsymbol{V}_{j})P(\boldsymbol{n}-\boldsymbol{V}_{j},t)-n_{i}f_{j}(\boldsymbol{n})P(\boldsymbol{n},t)).
\end{equation}
By applying the transformation $\boldsymbol{n}-\boldsymbol{V}_j\rightarrow\boldsymbol{n}$ in the first term of the right side of Eq.~\eqref{eq.mean1}, we obtain
\begin{align*}
\sum_{\boldsymbol{n}}n_i\frac{\partial P(\boldsymbol{n},t)}{\partial t}&=\sum_{j=1}^{M}\sum_{\boldsymbol{n}} ((n_{i}+V_{ij})f_{j}(\boldsymbol{n})P(\boldsymbol{n},t)-n_{i}f_{j}(\boldsymbol{n})P(\boldsymbol{n},t))\\
&= \sum_{j=1}^{M}\sum_{\boldsymbol{n}} V_{ij}f_{j}(\boldsymbol{n})P(\boldsymbol{n},t).
\end{align*}
Thus,
\begin{equation}\label{eq.mean2}
\frac{d\mu_{i}}{dt}=\sum_{j=1}^{M}V_{ij}\langle f_j(\boldsymbol{n})\rangle.
\end{equation}
Up till now, $\boldsymbol{n}$ is considered to be a vector of positive integers. To express the time derivative of the first moment by only itself and the second central moment, we assume that $\boldsymbol{n}$ is a vector of continuous real numbers and apply the Taylor expansion for $f_j(\boldsymbol{n})$ around $\boldsymbol{\mu}$ as follows:
\begin{equation}\label{eq.taylor}
f_j(\boldsymbol{n})=f_j(\boldsymbol{\mu})+(\boldsymbol{n}-\boldsymbol{\mu})^{\top}\frac{\partial f_j(\boldsymbol{\mu})}{\partial \boldsymbol{n}}+\frac{1}{2}(\boldsymbol{n}-\boldsymbol{\mu})^{\top}\frac{\partial^2 f_j(\boldsymbol{\mu})}{\partial \boldsymbol{n}^2}(\boldsymbol{n}-\boldsymbol{\mu}) + O(|\boldsymbol{n}-\boldsymbol{\mu}|^3).
\end{equation}
By approximating $f_j(\boldsymbol{n})$ to the order of $|\boldsymbol{n}-\boldsymbol{\mu}|^2$ and utilizing the fact that $\mathbb{E}[\boldsymbol{n}-\boldsymbol{\mu}]=0$, we obtain the following approximation:
\begin{align*}
\langle f_j(\boldsymbol{n})\rangle&=f_j(\boldsymbol{\mu})+\frac{1}{2}\sum_{h,l}\frac{\partial^2 f_j(\boldsymbol{\mu})}{\partial n_h\partial n_l}\langle (n_h-\mu_h)(n_l-\mu_l)\rangle\\
&=f_j(\boldsymbol{\mu})+\frac{1}{2}\sum_{h,l}\frac{\partial^2 f_j(\boldsymbol{\mu})}{\partial n_h\partial n_l}\sigma_{hl}.
\end{align*}
Substituting the above result into Eq.~\eqref{eq.mean2}, we obtain a differential equation of the first moment as follows:
\begin{equation}
\frac{d\mu_{i}}{dt}=\sum_{j=1}^{M}V_{ij}\left(f_j(\boldsymbol{\mu})+\frac{1}{2}\sum_{h,l}\frac{\partial^2 f_j(\boldsymbol{\mu})}{\partial n_h\partial n_l}\sigma_{hl}\right).
\end{equation}
Similarly, to obtain the equation of the second central moment, we multiply $(n_i-\mu_i)(n_{i'}-\mu_{i'})$  by Eq.~\eqref{eq.cme} and take the sum of all possible states $\boldsymbol{n}$ to get the following equation:
\begin{equation}\label{eq.var1}
\frac{d\sigma_{ii'}}{dt}=\sum_{j=1}^{M}\sum_{\boldsymbol{n}} ((n_i-\mu_i)(n_{i'}-\mu_{i'})f_{j}(\boldsymbol{n}-\boldsymbol{V}_{j})P(\boldsymbol{n} -\boldsymbol{V}_{j},t)-(n_i-\mu_i)(n_{i'}-\mu_{i'})f_{j}(\boldsymbol{n})P(\boldsymbol{n},t)).
\end{equation}
By applying the transformation $\boldsymbol{n}-\boldsymbol{V}_j\rightarrow\boldsymbol{n}$ in the first term of the right side of Eq.~\eqref{eq.var1}, we obtain
\begin{align*}
\frac{d\sigma_{ii'}}{dt}&=\sum_{j=1}^{M}\sum_{\boldsymbol{n}} ((n_i+V_{ij}-\mu_i)(n_{i'}+V_{i'j}-\mu_{i'})f_{j}(\boldsymbol{n})P(\boldsymbol{n},t) - (n_i-\mu_i)(n_{i'}-\mu_{i'})f_{j}(\boldsymbol{n})P(\boldsymbol{n},t))\\
&=\sum_{j=1}^{M}\sum_{\boldsymbol{n}}(V_{ij}V_{i'j}+V_{i'j}(n_i-\mu_i)+V_{ij}(n_{i'}-\mu_{i'}))f_{j}(\boldsymbol{n})P(\boldsymbol{n},t))\\
&=\sum_{j=1}^{M}(V_{ij}V_{i'j}\langle f_j(\boldsymbol{n})\rangle+V_{i'j}\langle (n_i-\mu_i)f_{j}(\boldsymbol{n})\rangle+V_{ij}\langle (n_{i'}-\mu_{i'})f_{j}(\boldsymbol{n})\rangle).
\end{align*}
Using the approximation of $f_j(\boldsymbol{n})$ in Eq.~\eqref{eq.taylor} and truncating all central moments of order higher than two, we get the following approximation:
\begin{equation*}
\langle(n_i-\mu_i)f_{j}(\boldsymbol{n})\rangle=\sum_{l}\frac{\partial f_j(\boldsymbol{\mu})}{\partial n_l}\langle(n_i-\mu_i)(n_l-\mu_l)\rangle=\sum_{l}\frac{\partial f_j(\boldsymbol{\mu})}{\partial n_l}\sigma_{il}.
\end{equation*}
Consequently, the equation of the second central moment is acquired as follows:
\begin{align*}
\frac{d\sigma_{ii'}}{dt}=\sum_{j=1}^{M}\left(V_{ij}V_{i'j}\left(f_j(\boldsymbol{\mu})+\frac{1}{2}\sum_{h,l}\frac{\partial^2 f_j(\boldsymbol{\mu})}{\partial n_h\partial n_l}\sigma_{hl}\right)+V_{i'j}\sum_{l}\frac{\partial f_j(\boldsymbol{\mu})}{\partial n_l}\sigma_{il}+V_{ij}\sum_{l}\frac{\partial f_j(\boldsymbol{\mu})}{\partial n_l}\sigma_{i'l}\right).
\end{align*}

\section{Exact Solve algorithm}
{
\begin{algorithm}[H]
\begin{algorithmic}[1]
\Require A regular system $\{\mathcal{P},\mathcal{Q}\}\subset\mathbb{R}[x_1,\dots,x_n]$
\Ensure The number of distinct real solutions of system $\{\mathcal{P}=0,\mathcal{Q}>0\}$
\State $rs_1\gets \mathrm{FindRoots}(P_1(x_1))$
\If {$rs_1=[r_{11},\dots,r_{h_{1}1}]$} $pts_1\gets [[r_{11}],\dots,[r_{h_{1}1}]]$ \textbf{else} $pts_1\gets [~]$
\EndIf
\State set $i\gets 1$
\Repeat
\If {$pts_i=[~]$} $pts\gets [~]$ and go to step \ref{step.end.exact.solve}
\Else
\State let $pts_i=[[r_{11},\dots,r_{1i}],\dots,[r_{h_{i}1},\dots,r_{h_{i}i}]]$
\State set $l\gets 1$
\State substitute $x_1=r_{l1},\dots,x_i=r_{li}$ into the system $\{\mathcal{P}=0,\mathcal{Q}>0\}$\label{step.init}
\State $rs_{i+1}^{(l)}\gets \mathrm{FindRoots}(P_{i+1}(x_{i+1}))$
\If {$rs_{i+1}^{(l)}=[~]$} $pts_{i+1}^{(l)}=[~]$
\Else
\State let $rs_{i+1}^{(l)}=[r_1,\dots,r_{u}]$
\State $pts_{i+1}^{(l)}\gets [[r_{l1},\dots,r_{li},r_1],\dots,[r_{l1},\dots,r_{li},r_{u}]];~ l\gets l+1$
\EndIf
\If {$l\leq h_i$} go to step \ref{step.init} \textbf{else} $pts_{i+1}\gets \bigcup_{i'=1}^{h_i}pts_{i+1}^{(i')};~ i\gets i+1$
\EndIf
\EndIf
\Until {$i=n$}
\If {$pts=[~]$}\label{step.end.exact.solve} \Return $0$
\Else ~\Return $m$ if $m$ members of $pts$ make $\mathcal{Q}>0$ true
\EndIf
\end{algorithmic}
\caption{ExactSolve}\label{alg.exactsolve}
\end{algorithm}
}

\section{Probabilistic Test algorithm}
If a polynomial system always has at least one solution in $\mathbb{C}^{n}$, then the resultant of the system will be zero. Exploiting this property, one can calculate the resultant and check whether the system has solutions. However, as computational complexity of the resultant grows fast when the number of variables increases, it leads to a heavy computation and poor time performance. Suppose that the resultant is $R(\boldsymbol{u})$, where $\boldsymbol{u}$ is the parameter. Then instead of calculating exact form of $R(\boldsymbol{u})$, we compute $R(\bar{\boldsymbol{u}})~(\mathrm{mod}~p)$, where $\bar{\boldsymbol{u}}$ is a rational value of the parameter and $p$ is an arbitrary prime number. From a practical viewpoint, calculating $R(\bar{\boldsymbol{u}})~(\mathrm{mod}~p)$ is more efficient than symbolic computation of $R(\boldsymbol{u})$. If $R(\boldsymbol{u})=0$ then we always obtain the result of zero, since $R(\bar{\boldsymbol{u}})~(\mathrm{mod}~p)=0$. A large value of $p$ may lead to the high probability of the elimination; however, it also reduces the time performance.
{
\begin{algorithm}[H]
\begin{algorithmic}[1]
\Require A polynomial system $\mathcal{P}\subset\mathbb{K}[x_1,\dots,x_n]$ and a prime number $p$
\Ensure Return true if $\{\mathcal{P}=0\}$ probably has at least one solution in $\mathbb{C}^n$, else return false
\State substitute a random rational value of the parameter $\boldsymbol{u}=\bar{\boldsymbol{u}}$ into $\mathcal{P}$
\If {$\mathcal{P}=\{P_1,P_2,\dots,P_m\},~m>n$}
\State set $i_n\gets n,i_m\gets m$
\While {$i_n>0$}
\State $ps\gets \{P_1,P_2,\dots,P_{i_m}\}$
\If {$ps$ contains nonzero constant} \Return false \EndIf
\State sort $ps$ according to ascending degree of $x_{i_n}$
\If {$\mathrm{deg}(ps[i_m],x_{i_n})>0$}
\State $h\gets$ minimum index such that $\mathrm{deg}(ps[h],x_{i_n})>0$
\State update
\begin{align*}
\hspace{2cm}P_{i}&=ps[i],~\forall 1\leq i<h,~P_{i_m}=ps[i_m],\\
\hspace{2cm}P_{i}&=\mathrm{res}(ps[i_m],ps[i],x_{i_n})~(\mathrm{mod}~p),~\forall h\leq i<i_m
\end{align*}
\EndIf
\State set $i_n\gets i_n-1,~i_m\gets i_m-1$
\EndWhile
\If {$\mathcal{P}$ contains nonzero constant} \Return false \EndIf
\EndIf
\State \Return true
\end{algorithmic}
\caption{ProbabilisticTest}\label{alg.probtest}
\end{algorithm}
}

\section{Sequential Decomposition algorithm}
This sequential decomposition algorithm is based on \textbf{RegSer} in \citep{Wang.2000}. Unlike \textbf{RegSer} which produces a list of regular systems, this algorithm returns a list of decomposed systems and a regular system which may be empty. Several processes are added to eliminate the systems that have no solutions. The prime number $p$ is set $p=3$ in \textbf{ProbabilisticTest}. The pseudocode of the algorithm is shown as follows.
{
\begin{algorithm}[H]
\caption{SequentialDecomposition}\label{alg.seqdecomp}
\begin{algorithmic}[1]
\Require A polynomial system $\{\mathcal{T},\mathcal{U}\}\subset\mathbb{K}[\boldsymbol{x]}$ and a positive integer number $n$
\Ensure Return $[\Phi,\Psi]$, where $\Phi$ is a list of decomposing systems and $\Psi$ is a regular system
\State set $\Phi\gets\varnothing,\Psi\gets\varnothing$
\For {$m=n,\dots,1$}
\State set $\mathcal{T}\gets\mathcal{T}\setminus\{0\},~ \mathcal{U}\gets\mathcal{U}\setminus(\mathbb{K}\setminus\{0\})$
\If {$\mathcal{T}\cap\mathbb{K}\neq\varnothing$ or $0\in\mathcal{U}$} go to \ref{step.done} \EndIf
\If {$\exists~ u[\boldsymbol{x}]\in\mathcal{U},~ t[\boldsymbol{x}] \in\mathcal{T}$ such that $\mathrm{prem}(u,t)=0$} go to \ref{step.done} \EndIf
\If {ProbabilisticTest($\mathcal{T},3$) is false} go to step \ref{step.done}
\EndIf
\If {$\mathcal{T}^{\langle m\rangle}=\varnothing$} go to step \ref{step.UkIsNotEmpty}\EndIf
\State make all polynomials in $\mathcal{T}$ and $\mathcal{U}$ to be square-free
\While {true}
\State let $P_2$ be an element of $\mathcal{T}^{\langle m\rangle}$ with minimal degree in $x_m$ and set
\begin{align*}
\hspace{1.5cm}\Phi &\gets\Phi\cup[\{\mathcal{T}~\setminus~\{P_2\}\cup \{\mathrm{ini}(P_2),\mathrm{red}(P_2)\},\mathcal{U},m\}]\\
\hspace{1.5cm}\mathcal{U} &\gets \mathcal{U}\cup\{\mathrm{ini}(P_2)\}
\end{align*}
\If {$|\mathcal{T}^{\langle m\rangle}|=1$} go to step \ref{step.UkAndP} \textbf{else} take a polynomial $P_1$ from $\mathcal{T}^{\langle m\rangle}\setminus\{P_2\}$
\EndIf
\State compute the s.r.s $H_2,\dots,H_r$ of $P_1$ and $P_2$ w.r.t $x_m$
\State set $I_i\gets\mathrm{lc}(H_i,x_m)$ for $2\leq i\leq r$
\If {$\mathrm{lv}(H_r)\prec x_m$} set $\overline{r}\gets r-1$ \textbf{else} set $\overline{r} \gets r$
\EndIf
\State set
\begin{align*}
\hspace{1.5cm}\Phi &\gets\Phi\cup[\{\mathcal{T}\setminus\{P_1,P_2\}\cup\{H_i,I_{i+1},\dots,I_r\},\mathcal{U}\cup\{I_i\},m\}~|~2\leq i\leq \overline{r}-1]\\
\hspace{1.5cm}\mathcal{T} &\gets \mathcal{T}\setminus\{P_1,P_2\}\cup\{H_r,H_{\overline{r}}\},\quad\mathcal{U}\gets\mathcal{U}\cup\{I_{\overline{r}}\}
\end{align*}
\EndWhile
\While {$\mathcal{U}^{\langle m\rangle}\neq\varnothing$ and $\mathrm{lv}(P_2)=x_m$}\label{step.UkAndP}
\State let $P_1$ be a polynomial in $\mathcal{U}^{\langle m \rangle}$, compute s.r.s $H_2,\dots,H_r$ of $P_1$ and $P_2$ w.r.t $x_m$
\State set $I_i\gets\mathrm{lc}(H_i,x_m)$ for $2\leq i\leq r$
\State set
{\begin{align*}
\hspace{1.5cm}\Phi &\gets \Phi \cup [\{\mathcal{T}\setminus\{P_2\}\cup\{\mathrm{pquo}(P_2,H_i,x_m),I_{i+1},\dots,I_r\},\mathcal{U}\cup\{I_i\},m\}~|~2\leq i\leq r-1]\\
\hspace{1.5cm}\mathcal{T} &\gets \mathcal{T}\setminus\{P_2\}\cup\{\mathrm{pquo}(P_2,H_r,x_m)\},\quad P_2\gets \mathrm{pquo}(P_2,H_i,x_m)
\end{align*}}
\If {$\mathrm{lv}(H_r)\prec x_m$} set $\mathcal{U}\gets\mathcal{U}\setminus\{P_1\}\cup\{I_r\}$ \textbf{else} set $\mathcal{U}\gets\mathcal{U}\cup\{I_r\}$
\EndIf
\EndWhile
\If {$\mathcal{U}^{\langle m \rangle}\neq\varnothing$}\label{step.UkIsNotEmpty}
\ForAll {$P_1\in\mathcal{U}^{\langle m\rangle}$}
\State set
\begin{align*}
\hspace{2cm}\Phi &\gets \Phi \cup [\{\mathcal{T}\cup\{\mathrm{ini}(P_1)\},\mathcal{U}\setminus\{P_1\}\cup\{\mathrm{red}(P_1)\},m\}]\\
\hspace{2cm}\mathcal{U}&\gets\mathcal{U}\cup\{\mathrm{ini}(P_1)\}
\end{align*}
\EndFor
\EndIf
\EndFor
\State set $\Psi\gets \{\mathcal{T},\mathcal{U}\}$
\State \Return $[\Phi,\Psi]$\label{step.done}
\end{algorithmic}
\end{algorithm}
}
\end{widetext}

\end{document}